\documentclass[11pt,english]{article}

\usepackage{algorithm}
\usepackage{booktabs}
\usepackage{graphicx}
\usepackage{amsmath}
\usepackage{amssymb}
\usepackage{latexsym}
\usepackage{crop}
\usepackage{algorithmic,algorithm}
\usepackage{multirow}
\usepackage{bm}
\usepackage{bbm}
\usepackage{enumerate}
\usepackage{url}
\usepackage{array}
\usepackage{paralist}
\usepackage{diagbox}
\usepackage{wasysym}
\usepackage{booktabs}
\usepackage{placeins}
\usepackage[dvipsnames,table,xcdraw]{xcolor}
\usepackage[colorlinks = true, pdfstartview = FitV, linkcolor = blue, citecolor = blue, urlcolor = blue]{hyperref}

\usepackage{listings}

\usepackage[utf8]{inputenc}
\usepackage{rotating}

\usepackage[capitalise]{cleveref}
\crefname{equation}{}{}
\crefname{figure}{Figure}{Figures}
\creflabelformat{equation}{\textup{(#2#1#3)}}
\crefname{assumption}{Assumption}{Assumptions}
\crefname{condition}{Condition}{Conditions}

\usepackage{xspace}
\newcommand\st{\textsuperscript{st}\xspace}

\renewcommand\th{\textsuperscript{th}\xspace}

\usepackage{fullpage}
\usepackage{multirow}

\usepackage[sort,nocompress]{cite}

\usepackage{arydshln}
\usepackage{enumitem}
\setlist[enumerate,1]{leftmargin=*,wide=0em, noitemsep,nolistsep, label = {\bfseries \arabic*.}}
\setlist[itemize,1]{leftmargin=*,wide=0em, noitemsep,nolistsep}

\newlist{abbrv}{itemize}{1}
\setlist[abbrv,1]{label=,labelwidth=1in,align=parleft,itemsep=0.1\baselineskip,leftmargin=!}

\usepackage{pifont}
\makeatletter
\newcommand*{\transpose}{%
	{\mathpalette\@transpose{}}%
}
\newcommand*{\@transpose}[2]{%
	\raisebox{\depth}{$\m@th#1\intercal$}%
}
\makeatother

\renewcommand {\AA}  { {\bm{A}} }

\newcommand {\hY}  { {\widehat{Y}} }

\newcommand{\R}{{\rm I\!R}}

\newcommand {\vv}  { {\bm v} }

\newcommand{\defeq}{\mathrel{\mathop:}=}

\newcommand*\xbar[1]{%
	\hbox{%
		\vbox{%
			\hrule height 0.5pt 
			\kern0.5ex
			\hbox{%
				\kern-0.1em
				\ensuremath{#1}%
				\kern-0.1em
			}%
		}%
	}%
} 

\definecolor{forestgreen}{rgb}{0.13, 0.55, 0.13}

\definecolor{amber}{rgb}{1.0, 0.75, 0.0}

\definecolor{bananayellow}{rgb}{.8, 0.6, 0}

\newcounter{comment}

\usepackage{amsthm}
\usepackage[framemethod=TikZ]{mdframed}

\mdfdefinestyle{theoremstyle}{%
	linewidth = 1pt,%
	roundcorner = 10pt,%
	leftmargin = 0,%
	rightmargin = 0,%
	backgroundcolor = cyan!3,%
	outerlinecolor = magenta!70!black,%
	splittopskip = \topskip,%
	ntheorem = true,%
}
\mdtheorem[style=theoremstyle]{claim}{Claim}

\newmdtheoremenv[%
linewidth = 1pt,%
roundcorner = 10pt,%
leftmargin = 0,%
rightmargin = 0,%
backgroundcolor = green!3,%
outerlinecolor = blue!70!black,%
splittopskip = \topskip,%
ntheorem = true,%
]{theorem}{Theorem}

\newmdtheoremenv[%
linewidth = 1pt,%
roundcorner = 10pt,%
leftmargin = 0,%
rightmargin = 0,%
backgroundcolor = green!3,%
outerlinecolor = blue!70!black,%
splittopskip = \topskip,%
ntheorem = true,%
]{corollary}{Corollary}

\newmdtheoremenv[%
linewidth = 1pt,%
roundcorner = 10pt,%
leftmargin = 0,%
rightmargin = 0,%
backgroundcolor = green!3,%
outerlinecolor = blue!70!black,%
splittopskip = \topskip,%
ntheorem = true,%
]{lemma}{Lemma}

\newmdtheoremenv[%
linewidth = 1pt,%
roundcorner = 10pt,%
leftmargin = 0,%
rightmargin = 0,%
backgroundcolor = yellow!3,%
outerlinecolor = blue!70!black,%
splittopskip = \topskip,%
ntheorem = true,%
]{definition}{Definition}

\newmdtheoremenv[%
linewidth = 1pt,%
roundcorner = 10pt,%
leftmargin = 0,%
rightmargin = 0,%
backgroundcolor = green!3,%
outerlinecolor = blue!70!black,%
splittopskip = \topskip,%
ntheorem = true,%
]{proposition}{Proposition}

\newmdtheoremenv[%
linewidth = 1pt,%
roundcorner = 10pt,%
leftmargin = 0,%
rightmargin = 0,%
backgroundcolor = green!3,%
outerlinecolor = blue!70!black,%
splittopskip = \topskip,%
ntheorem = true,%
]{condition}{Condition}

\newmdtheoremenv[%
linewidth = 1pt,%
roundcorner = 10pt,%
leftmargin = 0,%
rightmargin = 0,%
backgroundcolor = green!3,%
outerlinecolor = blue!70!black,%
splittopskip = \topskip,%
ntheorem = true,%
]{assumption}{Assumption}

\theoremstyle{definition}
\newmdtheoremenv[%
linewidth = 1pt,%
roundcorner = 10pt,%
leftmargin = 0,%
rightmargin = 0,%
backgroundcolor = blue!3,%
outerlinecolor = blue!70!black,%
splittopskip = \topskip,%
ntheorem = true,%
]{example}{Example}

\theoremstyle{definition}
\newmdtheoremenv[%
linewidth = 1pt,%
roundcorner = 10pt,%
leftmargin = 0,%
rightmargin = 0,%
backgroundcolor = red!3,%
outerlinecolor = blue!70!black,%
splittopskip = \topskip,%
ntheorem = true,%
]{remark}{Remark}

\usepackage{tikz}
\usepackage{xparse}

\NewDocumentCommand\DownArrow{O{2.0ex} O{black}}{%
	\mathrel{\tikz[baseline] \draw [<-, line width=0.5pt, #2] (0,0) -- ++(0,#1);}
}

\usepackage{listings} 

\definecolor{mygreen}{rgb}{0,0.6,0}
\definecolor{mygray}{rgb}{0.5,0.5,0.5}
\definecolor{mymauve}{rgb}{0.58,0,0.82}

\newcommand {\Ex} { {\mathbb E} }

\usepackage{dsfont}

\usepackage{amsmath}
\usepackage{amsfonts}
\usepackage{amssymb}
\usepackage{makeidx}
\usepackage{graphicx}
\usepackage{caption}
\usepackage{subcaption}
\usepackage{framed}
\usepackage{booktabs,array}
\usepackage{xcolor}

\usepackage{algorithmic}

\DeclareDocumentCommand{\inn}{mgO{\R}}
  {%
  \IfNoValueTF{#2}
    {\in #3^{#1}}
    {\in #3^{#1 \times #2}}%
}

\allowdisplaybreaks

\begin{document}
	\title{Rollage: Efficient Rolling Average Algorithm \\ to Estimate ARMA Models for Big Time Series Data}
	\author{
		Ali Eshragh\thanks{Carey Business School, Johns Hopkins Univeristy, MD, USA. Email:  \tt{ali.eshragh@jhu.edu}}\ \thanks{International Computer Science Institute, University of California at Berkeley, CA, USA.}
		\and
		Glen Livingston\thanks{School of Information and Physical Sciences, University of Newcastle, Australia. Emails:  \tt{glen.livingstonjr@newcastle.edu.au, thomas.mccarthymccann@uon.edu.au, Luke.Yerbury@uon.edu.au}} 
		\and 
		Thomas McCarthy McCann$^{\ddagger}$
		\and 
		Luke Yerbury$^{\ddagger}$
	}	
	\date{}
	\maketitle


\begin{abstract}
	We develop a new efficient algorithm for the analysis of large-scale time series data. We firstly define rolling averages, derive their analytical properties, and establish their asymptotic distribution. These theoretical results are subsequently exploited to develop an efficient algorithm, called \texttt{Rollage}, for fitting an appropriate AR model to big time series data. When used in conjunction with the Durbin's algorithm, we show that the \texttt{Rollage} algorithm can be used as a criterion to optimally fit \texttt{ARMA} models to big time series data. Empirical experiments on large-scale synthetic time series data support the theoretical results and reveal the efficacy of this new approach, especially when compared to existing methodology.   
\end{abstract}

	
	\section{Introduction}

A time series is a collection of random variables indexed according to the order they are obtained in time. The primary objective of time series analysis is to develop statistical models to forecast the future behavior of the system. These models have proved their effectiveness and advantages in modeling and analyzing stochastic dynamic systems, and continue to gain in popularity for modeling a wide range of applications spanning from supply chains  and energy systems to epidemiology and engineering problems \cite{Abolghasemi2020,Eshragh2019,Eshragh2020,Eshragh2022,Shu}.

The \emph{autoregressive moving average} (\texttt{ARMA}) model is a widely applied model to achieve this objective. The model was popularized by Box and Jenkins  \cite{Box} for analyzing stationary time series data. Extensions of this model were subsequently introduced, such as the \emph{autoregressive integrated moving average} (\texttt{ARIMA}) model for analyzing non-stationary time series data which posses a trend in mean, and the seasonal \texttt{ARIMA} (\texttt{SARIMA}) model to deal with time series data displaying seasonal effects \cite{Shu}. All in all, each \texttt{SARIMA} model is, at its core, an \texttt{ARMA} model for a linearly transformed time series constructed by differencing the original time series at proper lags. 

Broadly speaking, the Box-Jenkins method involves three steps, consisting of identification, estimation and diagnosis of an \texttt{ARMA} model. The first step, model identification, involves making an initial guess for an appropriate order of the model. Throughout the literature, there have been numerous procedures proposed to estimate the orders of an \texttt{ARMA} model \cite{CHOI}. Commonly the orders are chosen through use of the \emph{autocorrelation function} ($\mathsf{ACF}$) in conjunction with the \emph{partial autocorrelation function} ($\mathsf{PACF}$), and observing their respective plots. 

This method is not without its drawbacks as modelers require a high level of expertise to interpret the $\mathsf{ACF}$ and $\mathsf{PACF}$ plots manually \cite{Song&Esbogue}. In addition, the $\mathsf{PACF}$ only uses the last over-fitted coefficient to estimate the order of the \emph{autoregressive} (\texttt{AR}) component which may not truly incorporate all information from the sample. The main goal of this paper is to establish a new method to estimate the order of an \texttt{AR} model that utilizes all over-fitted coefficients in big data regime.

The second step of the Box-Jenkins method involves estimating the parameters for the chosen \texttt{ARMA} model. For pure \texttt{AR} models, this can be achieved analytically by performing a variant of maximum likelihood estimation ($\mathsf{MLE}$) known as conditional maximum likelihood estimation ($\mathsf{CMLE}$) \cite{Hamilton1994Book}. However,
for pure \texttt{MA} and combined \texttt{ARMA} models, both the likelihood and conditional likelihood are complex non-linear functions. Hence, finding a solution to either the $\mathsf{MLE}$ or $\mathsf{CMLE}$ is intractable. Consequently, numerical optimization methods are often employed to estimate the parameters of \texttt{MA} and \texttt{ARMA} models.

Durbin \cite{Durbin1959} worked to overcome the intractability of the maximum likelihood equations of \texttt{MA} models by exploiting the asymptotic equivalence between \texttt{AR} models of infinite order and \texttt{MA} models of finite order. The solution was to fit
a large order AR model to represent the infinite order \texttt{AR} model using  $\mathsf{CMLE}$. The residuals of that model can then be used as estimates of the unobservable white noise of the \texttt{MA} process. Standard linear regression techniques could then be used to estimate the parameters of the \texttt{MA} model. Durbin \cite{Durbin1960} then extended this initial work to estimate the parameters for a full \texttt{ARMA} model using a similar approach. Durbin's method is promising as it replaces a non-linear estimation problem with two stages of linear estimation \cite{Broersen1996}. However, a good question to ask is, how large should the order of the AR model be in Durbin's methodology to provide both accurate estimates and optimal efficiency?

Through practice and simulation, Broersen \cite{Broersen1996,Broersen2000} identified that increasing the order of the large \texttt{AR} model does not necessarily result in more accurate estimate of an \texttt{ARMA} model and that the optimal order is in fact finite. As a result, an appropriately large order is essential in Durbin's algorithm to ensure parameter accuracy and optimal forecasts, as well as improving computational complexity. The use of model selection criteria, such as the Akaike Information Criteria (AIC) \cite{AIC1974}, has been suggested to aide in the selection of an appropriately large order \texttt{AR} model \cite{Hannan1982,Broersen1996}.

With these developments serving as motivation, this paper will look to develop a new algorithm to appropriately estimate the order of an \texttt{AR} model, and then use this algorithm as a criterion in Durbin's methodology to optimally fit an \texttt{ARMA} model to big time series data. In particular, our contributions can be summarized as follows:

\begin{enumerate}[label=(\roman*)]
    \itemsep 0.5em 
    \item We introduce the concept of a rolling average and derive its theoretical properties,
    \item Using the rolling average properties, we develop a highly-efficient algorithm, called \texttt{Rollage}, for fitting an appropriate \texttt{AR} model to big time series data,
    \item We use the \texttt{Rollage} algorithm as a model selection criterion in Durbin's methodology to optimally fit an \texttt{ARMA} model,
    \item We empirically demonstrate the effectiveness of the \texttt{Rollage} algorithm to estimate \texttt{ARMA} models on large-scale synthetic time series data, when compared to existing criteria including $\mathsf{BIC}$ and $\mathsf{GIC}$. 
\end{enumerate}

The structure of this paper is as follows: \cref{sec:Background} introduces three times series models utilized in this paper, namely \texttt{AR}, \texttt{MA} and \texttt{ARMA} models, and covers their important properties and estimation techniques. \cref{sec:TheoreticalResults} introduces the concept of a rolling average, derives its theoretical results, and develops a new methodology to estimate \texttt{AR}, \texttt{MA}, and \texttt{ARMA} models, appropriately called \texttt{Rollage} algorithm. \cref{sec:empirical_results} illustrates the efficacy of the new methodology by implementing it on several synthetic big time series data and comparing it to existing methodology. \cref{sec:Conclusion} concludes the paper and addresses future work.

\subsubsection*{Notation}
Throughout the paper, vectors and matrices are denoted by bold lower-case, and upper-case letters respectively (e.g., $\vv$ and $M$). All vectors are assumed to be column vectors. We use regular lower-case to denote scalar constants (e.g., $ c $). Random variables are denoted by regular upper-case letters (e.g., $ X $). For a real vector, $ \vv $, its transpose is denoted by $ \vv^{\transpose} $. For a vector $\vv$ and a matrix $M$, $\|\vv\|$ and $\|M\|$ denote vector $\ell_{2}$ norm and matrix spectral norm, respectively. The determinant and adjugate of a square matrix $ M $ are denoted by $ \mathsf{det}(M) $ (or $ |M| $, used interchangeably) and $\mathsf{adj}(M)$, respectively. Adopting \texttt{Matlab} notation, we use $ A(i,:) $ and $ A(:,j) $ to refer to the $ i\th $ row and $ j\th $ column of the matrix $ A $, respectively, and consider them as a column vector.

	
\section{Background}
\label{sec:Background}

In this section, we present a brief overview of the three time series models considered in this paper, namely autoregressive models (\cref{sec:AR}), moving average models (\cref{sec:AR}), and autoregressive moving average models (\cref{sec:ARMA}).

\subsection{Autoregressive Models}
\label{sec:AR}

A time series $\{Y_t;\, t=0,\pm 1,\pm 2,\ldots\}$ is called (weakly) stationary, if the mean $\Ex[Y_t]$ is independent of time $t$, and the auto-covariance $\mathsf{Cov}(Y_{t},Y_{t+h})$, denoted by $\gamma_h$, depends only on the lag $h$ for any integer values $t$ and $h$. A stationary time series $\{Y_t;\, t=0,\pm 1,\pm 2,\ldots\}$  with the constant mean $\Ex[Y_t]=0$ is an \texttt{AR} model with the order $p$, denoted by $\mathtt{AR}(p)$, if we have
\begin{align}
	\label{eq:AR(p)}
	\medskip Y_t & = \phi_1^{(p)} Y_{t-1}+\cdots+\phi_p^{(p)} Y_{t-p}+W_t, 
\end{align}
where $\phi_p^{(p)} \neq 0$ and the time series $\{W_t;\, t=0,\pm 1,\pm 2,\ldots\}$ is a Gaussian white noise with the mean $\Ex[W_t] = 0$ and variance $\mathsf{Var}(W_t) = \sigma_{W}^2$. Gaussian white noise is a stationary time series in which each individual random variable $W_t$ has a normal distribution and any pair of random variables $W_{t_1}$ and $W_{t_2}$ for distinct values of $t_1,t_2\in\mathbb{Z}$ are uncorrelated. For the sake of simplicity, we assume that $\Ex[Y_t] = 0$

It is readily seen that each $\mathtt{AR}(p)$ model has $p+2$ unknown parameters consisting of the order $p$, the coefficients $\phi_i^{(p)}$ and the variance of white noises $\sigma_W^2$. Following is a brief explanation of a common method in the literature for estimating the unknown order $p$.

\paragraph{Estimating the order $\bm{p}$.} A common method to estimate the order of an $\mathtt{AR}(p)$ model is to use the \emph{partial autocorrelation function} (\textsf{PACF}) \cite[Chapter 3]{Shu}. The \textsf{PACF} of a stationary time series $\{Y_t;\, t=0,\pm 1,\pm 2,\ldots\}$ at lag $m$ is defined by
\begin{align}
	\mathtt{PACF}_m & \defeq \begin{cases}
		\medskip \rho(Y_t,Y_{t+1}) & \mbox{for}\ m=1, \\
		\medskip \rho(Y_{t+m}-\hY_{t+m,-m}, Y_t-\hY_{t,m}) & \mbox{for}\ m\geq 2,
	\end{cases}
	\label{eq:PACF}
\end{align}
where $\rho$ denotes the correlation function, and where $\hY_{t,m}$ and $\hY_{t+m,-m}$ denote the linear regression, in the population sense, of $Y_t$ and $Y_{t+m}$ on $\{Y_{t+1},\ldots, Y_{t+m-1}\}$, respectively. In order to apply the \textsf{PACF} values to estimate the order of an \texttt{AR} model, we first need to introduce the concept of ``causality'', as given in \cref{def:causal_AR}. 

\begin{definition}[Causal \texttt{AR} Model]\label{def:causal_AR}
	An $\mathtt{AR(p)}$ model is said to be ``causal'', if the time series $\{Y_t;\, t=0,\pm 1,\pm 2,\ldots\}$ can be written as 
	\begin{align*}
		Y_t & = \sum_{i=0}^{\infty} \psi_i W_{t-i},
	\end{align*}	
	where $\psi_0 = 1$ and the constant coefficients $\psi_i$ satisfy $\sum_{i=0}^{\infty}|\psi_i|<\infty$. 
\end{definition}

It can be shown that for a causal $\mathtt{AR}(p)$ model, while the theoretical \textsf{PACF} \cref{eq:PACF} at lags $m=1,\ldots,p-1$ may be non-zero and at lag $m=p$ is strictly non-zero, at lag $m=p+1$ it drops to zero and then remains at zero henceforth \cite[Chapter 3]{Shu}. \cref{thm:lim_dist_MLE} indicates the statistical properties of the parameter estimates using \textsf{PACF} and the \textsf{PACF} estimates. These can be used to select the model order in practice by plotting the sample \textsf{PACF} versus lag $m$ along with a $95\%$ zero-confidence boundary, that is two horizontal lines at ${\pm1.96}/{\sqrt{n}}$, are plotted. Then, the largest lag $m$ in which the sample \textsf{PACF} lies out of the zero-confidence boundary for \textsf{PACF} is used as an estimation of the order $p$.

Theorem \ref{thm:lim_dist_MLE} plays a crucial role in developing theoretical results in \cref{sec:TheoreticalResults}.

\begin{theorem}[Asymptotic Distribution of Estimated Coefficients \cite{Brockwell2009TimeSeries}]\label{thm:lim_dist_MLE}
	Suppose the time series $\{Y_1,\cdots, Y_n\}$ be a stationary casual $\mathtt{AR}(p)$ model as given in \cref{eq:AR(p)} and fit an $\mathtt{AR}(m)$ model ($m>p$) to the time series data, that is
	\begin{align*}
		\medskip Y_t & = \phi_1^{(m)}Y_{t-1} + \cdots + \phi_m^{(m)}Y_{t-m} + W_t. 
	\end{align*}
	The maximum likelihood estimate of the coefficient vector, denoted by \newline $ \hat{\bm{\phi}}_{p,m} = \begin{bmatrix} \hat{\phi}_1^{(m)} & \cdots & \hat{\phi}_m^{(m)} \end{bmatrix}^{\transpose} $, asymptotically, has a multivariate normal distribution
	\begin{align*}
		\medskip \sqrt{n} (\hat{\bm{\phi}}_{p,m} - \bm{\phi}_{p,m}) & \sim \mathtt{MN}(\bm{0}, \Sigma_{p,m}),
	\end{align*}
	where $\bm{\phi}_{p,m} := \begin{bmatrix} \phi_1^{(p)} & \cdots & \phi_p^{(p)} & 0 & \cdots & 0 \end{bmatrix}^{\transpose}$, the covariance matrix $\Sigma_{p,m} = \sigma_W^2 \Gamma_{p,m}^{-1}$, and 
	\begin{align*}
		\medskip \Gamma_{p,m} & = \begin{pmatrix}
			\gamma_0 & \gamma_1 & \cdots & \gamma_{m-1} \\
			\gamma_1 & \gamma_0 & \cdots & \gamma_{m-2} \\
			\vdots   & \vdots & \ddots & \vdots \\
			\gamma_{m-1} & \gamma_{m-2} & \cdots & \gamma_{0}
		\end{pmatrix},
	\end{align*}
	is the autocovariance matrix of the given time series. Furthermore, $\hat{\phi}_m^{(m)}$ is an unbiased estimate for $\mathtt{PACF}_m$ with the limit distribution
	\begin{align*}
		\medskip \sqrt{n}\hat{\phi}_m^{(m)} & \sim \mathtt{N}[0,1].
	\end{align*} 
\end{theorem}

\subsection{Moving Average Models}
\label{sec:MA}

A stationary time series $\{Y_t;\, t=0,\pm 1,\pm 2,\ldots\}$ with the constant mean $\Ex[Y_t]=0$ is an \texttt{MA} model with the order $q$, denoted by $\mathtt{MA}(q)$, if we have
\begin{align}
	\label{eq:MA(q)}
	\medskip Y_t & = \theta_1 W_{t-1} + \cdots + \theta_q W_{t-q} + W_t, 
\end{align}
where $\theta_q \neq 0$ and the time series $\{W_t;\, t=0,\pm 1,\pm 2,\ldots\}$ is a Gaussian white noise with the mean $\Ex[W_t] = 0$ and variance $Var(W_t) = \sigma_{W}^2$.

Similar to an $\mathtt{AR}(p)$ model, each $\mathtt{MA}(q)$ model has $q+2$ unknown parameters consisting of the order $q$, the coefficients $\theta_i$ and the variance of white noises $\sigma_W^2$. Here, we briefly explain the common methods in the literature to estimate each of these unknown parameters.

\paragraph{Estimating the order $\bm{q}$.} A common method to estimate the order of an $\mathtt{MA}(q)$ model is to use the \emph{autocorrelation function} (\textsf{ACF}) \cite[Chapter 3]{Shu}. The \textsf{ACF} of a stationary time series $\{Y_t;\, t=0,\pm 1,\pm 2,\ldots\}$ at lag $m$ is defined by
\begin{align}
	\mathtt{ACF}_m & \defeq \begin{cases}
		\medskip \displaystyle{\frac{\gamma_m}{\gamma_0}} & \mbox{for}\ m=0,1,\ldots, \\
		\medskip \mathtt{ACF}_{-m} & \mbox{for}\ m=-1,-2,\ldots.
		\end{cases}
		\label{eq:ACF}
\end{align}
In order to apply the \textsf{ACF} values to estimate the order of an \texttt{MA} model, we first need to introduce the concept of `` invertibility'', as given in \cref{def:invertible_MA}. 

\begin{definition}[Invertible \texttt{MA} Model]\label{def:invertible_MA}
	An $\mathtt{MA(q)}$ model is said to be ``invertible'', if embedded white noises of the time series $\{Y_t;\, t=0,\pm 1,\pm 2,\ldots\}$ can be written as 
	\begin{align*}
		W_t & = \sum_{i=0}^{\infty} \pi_i Y_{t-i},
	\end{align*}	
	where $\pi_0 = 1$ and the constant coefficients $\pi_i$ satisfy $\sum_{i=0}^{\infty}|\pi_i|<\infty$. 
\end{definition}

It can be shown that for an invertible $\mathtt{MA}(q)$ model, while the theoretical \textsf{ACF} \cref{eq:ACF} at lags $m=1,\ldots,q-1$ may be non-zero and at lag $m=q$ is strictly non-zero, at lag $m=q+1$ it drops to zero and then remains at zero henceforth \cite[Chapter 3]{Shu}. Furthermore, it can be shown that if $y_1,\ldots,y_n$ is a time series realization of an invertible $\mathtt{MA}(q)$ model, the (estimated) sample \textsf{ACF} values, scaled by $\sqrt{n}$, at lags greater than $q$ has a standard normal distribution, in limit. Thus, in practice, the (estimated) sample \textsf{ACF} versus lag $m$ along with a $95\%$ zero-confidence boundary, that is two horizontal lines at ${\pm1.96}/{\sqrt{n}}$, are plotted. Then, the largest lag $m$ in which the sample \textsf{ACF} lies out of the zero-confidence boundary for \textsf{ACF} is used as an estimation of the order $q$.

\paragraph{Maximum likelihood estimation of the coefficients $\bm{\theta_i}$ and variance $\bm{\sigma_{W}^2}$.} Unlike an \texttt{AR} model, for an $\mathtt{MA}(q)$ model both the log-likelihood function and the conditional log-likelihood function are complicated non-convex functions and cannot be maximized analytically \cite[Chapter 5]{Hamilton1994Book}. So, one approach in estimating the parameters of an \texttt{MA} model is maximizing the corresponding (log-)likelihood function approximately by applying some numerical optimization algorithms, such as the gradient descent method. 

The main difficulty in dealing with the likelihood function of an \texttt{MA} model is that the lagged values of white noises are not known before fitting a model to the data. \cite{Durbin1959} overcame this problem by developing a new method for \texttt{MA} model fitting. Motivated from \cref{def:invertible_MA}, it is readily seen that an ``invertible'' $\mathtt{MA}(q)$ model can be represented by
\begin{align*}
	\medskip Y_t = \sum_{i=1}^{\infty} (-\pi_i) Y_{t-1} + W_t,
\end{align*}
implying that an invertible $\mathtt{MA}(q)$ model is equivalent to an $\mathtt{AR}(\infty)$ model. Durbin exploited this equivalence to estimate the parameters of an \texttt{MA} model and later extended this for \texttt{ARMA} models.  As a result, \texttt{Durbin's Algorithm} first fits an \texttt{AR} model with a sufficiently large order, say $\tilde{p}$, to the data and approximates the values of white noises $W_t$ by finding the residuals of the fitted \texttt{AR} model, that is,
\begin{align}
	\label{eq:tilda_wt}
	\medskip \widetilde{w}_t & = y_t - \sum_{i=1}^{\tilde{p}} (-\hat{\pi}_i) y_{t-i}.
\end{align}
In the next step, the algorithm approximately estimates the coefficients $\theta_i$ by regressing the time series over the $q$ lagged values of the estimated residuals in \cref{eq:tilda_wt}, that is
\begin{align}
	\label{eq:regress_tilde_wt}
	\medskip Y_t & = \tilde{\theta}_1 \tilde{w}_{t-1} + \cdots + \tilde{\theta}_q \tilde{w}_{t-q} + W_t.
\end{align}
The steps of \texttt{Durbin's Algorithm} are depicted in \cref{alg:hr}. \cite{Hannan1982} extended this algorithm with trimming steps to improve the initial parameter estimates, as well as suggesting the use the \emph{Bayesian Information Criterion} (BIC) as a model selection criterion. Details regarding the BIC can be seen in \cref{def:BIC}. Both algorithms are readily applicable to ARMA models with the inclusion of the appropriate lags of the primary time series in \emph{Step 3}. 

\begin{definition}[Bayesian Information Criterion \cite{Shu}]\label{def:BIC}
	Consider a regression model with $k$ coefficients and denote the \emph{maximum likelihood estimator} for the variance as
	\begin{align*} 
	    \medskip \hat{\sigma}_k^2 = \frac{SSE(k)}{n} 
	\end{align*}
	where $SSE(k)$ is the residual sum of squares under the model with k regression coefficients. The \emph{Bayesian Information Criterion}, denoted $BIC$, is
	\begin{align*}
	    \medskip BIC = \log(\hat{\sigma_k^2}) + \frac{k\log(n)}{n}, 
	\end{align*}
	with the value of $k$ yielding the minimum $BIC$ specifying the best regression model.
\end{definition}

\begin{algorithm}[!h]
	\caption{\texttt{Durbin's Algorithm}: An Algorithm for \texttt{MA} Fitting}	
	\begin{algorithmic}
		\STATE \textbf{Input:} 
		\begin{itemize}[label=-]
			\vspace{2mm}
			\item Time series data $\{y_1,\ldots,y_n\}$\,;
			\vspace{2mm}
			\item The estimated order $q$\,;
		\end{itemize}
		\vspace{2mm}
		\item \emph{Step 1}. Choose a sufficiently high order $\tilde{p} > q$;
		\vspace{2mm}
		\item \emph{Step 2}. Find the \textsf{CMLE} of the coefficient vector, $(\hat{\pi}_{1},\ldots,\hat{\pi}_{\tilde{p}})$ and compute the white noise approximations $\widetilde{w}_t$ as in \cref{eq:tilda_wt};
		\vspace{2mm}
		\item \emph{Step 3}. Regress the primary time series over q lagged values of $\widetilde{w}_t$ to estimate the coefficient vector, $(\tilde{\theta}_{1},\ldots,\tilde{\theta}_{q})$ as in \cref{eq:regress_tilde_wt};
		\vspace{2mm}
		

		\STATE \textbf{Output:} Estimated coefficients $(\hat{\tilde{\theta}}_1,\ldots,\hat{\tilde{\theta}}_q)$.
	\end{algorithmic}
	\label{alg:hr}
\end{algorithm}

\cite{Broersen2000,Broersen1996} identified, through practice and simulation, that simply increasing the size of $\tilde{p}$ in the intermediate \texttt{AR} model does not necessarily improve the quality of estimates of the associated \texttt{MA} model. Broerson \cite{Broersen1996,Broersen2000} also identified that the optimal $\tilde{p}$ value is in fact finite and advocated the use of model selection criteria, such as the AIC \cite{AIC1974}, in order to choose an appropriate $\tilde{p}$. Instead, Broerson \cite{Broersen1996,Broersen2000} developed a new selection criterion for optimally choosing $\tilde{p}$ by generalizing the AIC (GIC), details of which can be seen in \cref{def:GIC}.
\begin{definition}[Generalized Information Criterion \cite{Broersen1996}]\label{def:GIC}
	Consider an $\mathtt{MA}(q)$ process as defined in \cref{eq:MA(q)} with variance
	\begin{align*}
	    \medskip \sigma_Y^2 = \sigma_W^2\left(1+\sum_{i=1}^{q}\theta_i^2\right) = \frac{\sigma_W^2}{{\prod_{i=1}^{\infty}}(1-k_i^2)}
	\end{align*}
	where $\theta_i$ are the coefficients of an $\mathtt{MA}$ model and $k_i$ are the reflection coefficients of a long $\mathtt{AR}$ model. Then the residual sum of squares for N observations of an $\mathtt{MA}(q)$ process, denoted $RSS_p$, is given by
	\begin{align*}
	    \medskip RSS_p =\sum_{n=1}^{N}\sigma_Y^2 \prod_{i=1}^{p}(1-k_i^2) = RSS_{p-1}(1-k_p^2).
	\end{align*}
	The \emph{Generalized Information Criterion}, denoted $GIC(p,\alpha)$, with penalty factor 1 for $\alpha$, is
	\begin{align*}
	    \medskip GIC(p, \alpha) = \log
	    \left(\frac{RSS_p}{N}\right) + \frac{\alpha p}{N}
	\end{align*}
	where the optimal order $k$ is the order with minimum $GIC(p,1), p=0,1,\ldots, \infty$.
\end{definition}

\subsection{Autoregressive Moving Average Models}
\label{sec:ARMA}

A stationary time series $\{Y_t;\, t=0,\pm 1,\pm 2,\ldots\}$ with the constant mean $\Ex[Y_t]=0$ is an \texttt{ARMA} model with orders $p$ and $q$, denoted by $\mathtt{ARMA}(p,q)$, if we have
\begin{align}
	\label{eq:AR(p)_1}
	\medskip Y_t & = \phi_1 Y_{t-1} + \cdots + \phi_p Y_{t-p} + W_t + \theta_1 W_{t-1} + \cdots + \theta_q W_{t-q}, 
\end{align}
where $\phi_p \neq 0 $, $\theta_q \neq 0$ and $W_t$ is Gaussian white noise with the mean $\Ex[W_t] = 0$ and variance $Var(W_t) = \sigma_{W}^2$. It is readily seen that an \texttt{ARMA} model is a combination of the \texttt{AR} and \texttt{MA} model seen previously in \cref{sec:AR} and \cref{sec:MA} respectively. Consequently, each $\mathtt{ARMA}(p,q)$ model has $p+q+3$ unknown parameters consisting of orders $p$ and $q$, the coefficients $\phi_i$ and $\theta_i$ and the variance of the white noises $\sigma_{W}^2$. Here, we briefly explain the common methods in the literature to estimate each of these unknown parameters.

\paragraph{Estimating the orders $\bm{p}$ and $\bm{q}$} Since each \texttt{ARMA} model contains both \texttt{AR} and \texttt{MA} components, the order of an $\mathtt{ARMA}(p,q)$ model is commonly estimated by using the $\mathsf{PACF}$ seen in \cref{sec:AR} in conjunction with the $\mathsf{ACF}$ seen in \cref{sec:MA}. \cref{tab:Behaviour_ACF_PACF} \cite{Shu} summarises the behaviour of the $\mathsf{ACF}$ and $\mathsf{PACF}$ for \texttt{ARMA} models and its derivatives. For pure \texttt{AR} and \texttt{MA} models, the behaviour of the $\mathsf{PACF}$ and $\mathsf{ACF}$ is clear and estimation of $p$ or $q$ is straight forward as previously discussed in \cref{sec:AR} and \cref{sec:MA} respectively. However, the behaviour of the $\mathsf{PACF}$ and $\mathsf{ACF}$ for an $\mathtt{ARMA}$ model is ambiguous with no clear cut off after a particular lag to estimate the orders $p$ and $q$. As a result, a more precise inspection of the $\mathsf{PACF}$ and $\mathsf{ACF}$ is required, and often multiple models, consisting of different combinations of orders $p$ and $q$, are made as an initial guess of what the true orders may be. Steps 2 and 3 of the Box-Jenkins method, estimation and diagnosis, are then performed to help distinguish the best modelling combination of $p$ and $q$.

\begin{table}[htb]
\centering
\medskip
\begin{tabular}{ccccccc}
\hline
\rowcolor[HTML]{EFEFEF} 
 &  \textbf{$\mathtt{AR}(p)$} & \textbf{$\mathtt{MA}(q)$} & \textbf{$\mathtt{ARMA}(p,q)$}  \\ \hline
\cellcolor[HTML]{EFEFEF}\textbf{$\mathsf{ACF}$} & \textsf{Tails off} & \textsf{Cuts off after lag $q$} & \textsf{Tails off} \\ \hline
\cellcolor[HTML]{EFEFEF}\textbf{$\mathsf{PACF}$} & \textsf{Cuts off after lag $p$} & \textsf{Tails off} & \textsf{Tails off} \\ \hline
\end{tabular}
\caption{This table highlights the behaviour of the $\textsf{ACF}$ and $\textsf{PACF}$ for $\mathtt{AR}(p)$, $\mathtt{MA}(q)$ and $\mathtt{ARMA}(p,q)$ processes.}
\label{tab:Behaviour_ACF_PACF}
\end{table}

\paragraph{Maximum likelihood estimation of the coefficients $\bm{\phi_i}$ and $\bm{\theta_i}$ and variance $\bm{\sigma_{W}^2}$:} Due to the \texttt{MA} component, both the log-likelihood and conditional log-likelihood function of an $\mathtt{ARMA}(p,q)$ model are complicated non-convex functions and cannot be maximized analytically. Numerical optimization techniques are commonly employed to deal with estimating the coefficients of \texttt{ARMA} models. Durbin's methodology, \cref{alg:hr}, of exploiting the asymptotic equivalence between $\mathtt{AR}(\infty)$ and $\mathtt{MA}(q)$ models can similarly be used to estimate the coefficients of \texttt{ARMA} models. The optimal order $\tilde{p}$ can also be chosen using model selection criteria such as the $\mathsf{BIC}$ utilised by Hannan and Rissanen \cite{Hannan1982} and the $\mathsf{GIC}$ utilised by \cite{Broersen1996,Broersen2000} as seen in \cref{sec:MA}. Durbin's methodology, \cref{alg:hr}, along with the $\mathsf{BIC}$ and  $\mathsf{GIC}$ criteria will be used as the estimation method for \texttt{ARMA} models for the remainder of this paper, and will be compared empirically in \cref{sec:empirical_results} to the new algorithm, called \texttt{Rollage}, developed in \cref{sec:TheoreticalResults}.

	
	\section{Theoretical Results}\label{sec:TheoreticalResults}

In this section, we introduce the concept of \emph{rolling average} and develop its theoretical properties. These results are utilized to construct an efficient algorithm to estimate an appropriate \texttt{AR} model for big time series data. It should be noted that all proofs of this section as well as technical lemmas/propositions/theorems used in the proofs are presented in the Supplementary Materials document.

\subsection{Rolling Average}
\label{sec:RA}

We start this section by introducing the \emph{rolling average} which is plays a central role in this work. 

\begin{definition}[Rolling Average]\label{def:rolling_average}
	Suppose the time series $\{Y_1,\cdots, Y_n\}$ is a causal $\mathtt{AR}(p)$ process. Fit an $\mathtt{AR}(m)$ model ($m > p$) to the data and find the \textsf{MLE} of the coefficient vector, $\hat{\bm{\phi}}_{p,m} = (\hat{\phi}_{1}^{(m)},\cdots,\hat{\phi}_{m}^{(m)})^\transpose$. The ``rolling average'' of this estimation is denoted by $\bar{\phi}_{p,m}$ and defined as follows: 
	\begin{align*}
	\medskip \bar{\phi}_{p,m} & := \frac{1}{m-p} \displaystyle{\sum_{j=p+1}^{m}\hat{\phi}_{j}^{(m)}}.
	\end{align*}
\end{definition}

The main motivation of introducing rolling averages is the convergence result on the asymptotic distribution of the \textsf{MLE}s of an \texttt{AR} model coefficients fitted to (large enough) time series data, as expressed in \cref{thm:lim_dist_MLE}. More precisely, \cref{thm:lim_dist_MLE} states that if the underlying model of the time series data is truly an $\mathtt{AR}(p)$ model, but an $\mathtt{AR}(m)$ model ($m > p$) is (over)fitted to the data, the \textsf{MLE} of the (over-fitting) coefficient vector, asymptotically, has a multivariate normal distribution, that is,
\begin{align*}
	\medskip \sqrt{n}\begin{bmatrix} \hat{\phi}_{p+1}^{(m)} & \cdots & \hat{\phi}_{m}^{(m)} \end{bmatrix}^{\transpose} & \sim \mathtt{MN}(\bm{0}, \Sigma_{p,m}(p+1,m:p+1:m)).
\end{align*}

Accordingly, the rolling average $\bar{\phi}_{p,m}$ which is the sample mean of the estimated (over-fitting) coefficients $\hat{\phi}_{p+1}^{(m)},\cdots,\hat{\phi}_{m}^{(m)}$ should also have a Normal distribution with a mean zero, in limit. This implies that, similar to the \textsf{PACF} values, the rolling averages can be also considered as a tool to estimate the order of an \texttt{AR} model. Furthermore, while the former approach looks only at the last (over-fitting) coefficient $\hat{\phi}_{m}^{(m)}$ to estimate the order of the model, rolling averages consider all (over-fitting) coefficients together. Hence, one could expect that the latter exploits more information from the sample, and consequently may provide more efficient and accurate estimates. This is the motivation behind developing a new algorithm for estimating an \texttt{AR} model based on rolling averages. 

Hence, we should derive the variance of the rolling averages to be able to utilize them in an algorithmic way to estimate the order of an \texttt{AR} model. For this purpose, we first define \emph{nested lower right corner matrix} in \cref{def:NLRC_matrix} and then obtain its structure in \cref{thm:lrc_matrix}. Finally, \cref{thm:lim_dist_rolling_average} applies all these results to establish the asymptotic distribution of the rolling average $\bar{\phi}_{p,m}$. 

\begin{definition}[Nested Lower Right Corner Matrix (NLRC)]\label{def:NLRC_matrix}
	Suppose the time series $\{Y_1,\cdots, Y_n\}$ is a causal $\mathtt{AR}(p)$ process. Fit an $\mathtt{AR}(m)$ model ($m > p$) to the data and find the covariance matrix $\Sigma_{p,m}$. The $(m-p)\times (m-p)$ square matrix $\Sigma_{p,m}(p+1:m,p+1:m)$ extracted from the lower right corner of the covariance matrix is called ``nested lower right corner'' matrix and denoted by $\mathsf{NLRC}_{p,m}$.
\end{definition}

\begin{theorem}[Closed Form for $\mathsf{NLRC}_{p,m}$] \label{thm:lrc_matrix}
	Let the time series $\{Y_1,\cdots, Y_n\}$ be a causal $\mathtt{AR}(p)$ model. The nested lower right corner matrix $\mathsf{NLRC}_{p,m}$ for a fixed $p \geq 1$ and $m=p+1,p+2,\ldots$, is a symmetric positive definite matrix with the lower triangular coordinates satisfying the following ``nested'' equations for  $m=p+2,p+3,\ldots$,
	\begin{align*}
		\medskip \mathsf{NLRC}_{p,m}(i,j) & = \begin{cases}
			\medskip \underline{\mbox{if } m \leq 2p+1:} \\
			\medskip \phi_{m-p-i}^{(p)}\phi_{m-p-1}^{(p)} + \mathsf{NLRC}_{p,m-1}(i,1) & \mbox{for } j=1,\ i=1,\ldots,m-p-1 \\
			\medskip \phi_{0}^{(p)}\phi_{m-p-1}^{(p)} & \mbox{for } j=1,\ i=m-p \\
			\medskip \underline{\mbox{if } m \geq 2p+2:} \\
			\medskip \mathsf{NLRC}_{p,m-1}(i,1) & \mbox{for } j=1,\ i=1,\ldots,m-p-1 \\
			\medskip 0 & \mbox{for } j=1,\ i=m-p \\
			\medskip \underline{\mbox{for all } m \geq p+2:} \\
			\medskip \mathsf{NLRC}_{p,m-1}(i-1,j-1) & \mbox{for } 2 \leq j \leq i \leq m-p,
			\end{cases}
	\end{align*}
	where $\phi_0^{(p)} := -1$ and the initial condition is  $\mathsf{NLRC}_{p,p+1}(1,1) = (\phi_0^{(p)})^2$. Moreover, the lower triangular coordinates equal
	\begin{align*}
		\medskip \mathsf{NLRC}_{p,m}(i,j) & = \begin{cases}
			\medskip \displaystyle{\sum_{k=0}^{\min\{ m-p-i, p \}}(\phi_k^{(p)})^2} & \mbox{for } j=i,\ i=1,\ldots,m-p \\
			\medskip \displaystyle{\sum_{k=0}^{\min\{ m-p-i, p-1 \}}\phi_k^{(p)} \phi_{k+1}^{(p)}} & \mbox{for } j=i-1,\ i=2,\ldots,m-p \\
			\medskip \hspace*{1.1cm} \vdots & \hspace*{1.1cm} \vdots \\
			\medskip \displaystyle{\sum_{k=0}^{\min\{ m-p-i, p-\ell \}}\phi_k^{(p)} \phi_{k+\ell}^{(p)}} & \mbox{for } j=i-\ell,\ i=\ell+1,\ldots,m-p \\
			\medskip \hspace*{1.1cm} \vdots & \hspace*{1.1cm} \vdots \\
			\medskip \hspace*{0.8cm} \displaystyle{\phi_0^{(p)} \phi_{p}^{(p)}} & \mbox{for } j=i-p,\ i=p+1,\ldots,m-p \\
			\medskip \hspace*{1.0cm} 0 & \mbox{elsewhere},			
		\end{cases}
	\end{align*}
	and above the diagonal coordinates are calculated through
	\begin{align*}
		\medskip \mathsf{NLRC}_{p,m}(i,j) & = \mathsf{NLRC}_{p,m}(j,i)\quad \mbox{for } 1 \leq i < j \leq m-p.
	\end{align*}
	Note that those coordinates where the corresponding range for $i$ has a lower bound greater than the upper bound should be disregarded.
\end{theorem}

\begin{example}
	For an $\mathtt{AR}(2)$ model, the lower triangular coordinates of the symmetric matrix $\mathsf{NLRC}_{2,7}$ are illustrated below (for sake of simplicity, the superscript $(2)$ has been omitted): 	
	\begin{align*}
		\medskip \mathsf{NLRC}_{2,7} & = \begin{pmatrix}
			\medskip \phi_0^2+\phi_1^2+\phi_2^2 & & & & & & \\
			\medskip \phi_0\phi_1 +\phi_1\phi_2 & \phi_0^2+\phi_1^2+\phi_2^2 & & & & \\ 
			\medskip \phi_0\phi_2 & \phi_0\phi_1 +\phi_1\phi_2 & \phi_0^2+\phi_1^2+\phi_2^2 	& & & \\ 
			\medskip 0 & \phi_0\phi_2 & \phi_0\phi_1 +\phi_1\phi_2 & \phi_0^2+\phi_1^2 & & \\ 
			\medskip 0 & 0 & \phi_0 \phi_2 & \phi_0 \phi_1 & \phi_0^2 \\
		\end{pmatrix}.
	\end{align*}
\end{example}

\begin{theorem}[Asymptotic Distribution of Rolling Average]\label{thm:lim_dist_rolling_average}
	Let the time series $\{Y_1,\cdots, Y_n\}$ be a causal $\mathtt{AR}(p)$ process. If an $\mathtt{AR}(m)$ model ($m > p$) is fitted to the data, asymptotically, we have
	\begin{align*}
		\medskip \sqrt{n}\bar{\phi}_{p,m} & \sim N[0, \sigma_{p,m}^2],
	\end{align*}
	where $\sigma_{p,m}^2$ satisfies the recursion, 
	\begin{align*}
		(m-p)^2\sigma_{p,m}^2 & = \begin{cases}
			(m-p-1)^2\sigma_{p,m-1}^2 + (\phi_0^{(p)} + \cdots + \phi_{p}^{(p)})^2 & \mbox{for } m=2p+1,2p+2,\ldots, \\
			(m-p-1)^2\sigma_{p,m-1}^2 + (\phi_0^{(p)} + \cdots + \phi_{m-p-1}^{(p)})^2 & \mbox{for } m=p+2,\ldots,2p, \\
			(\phi_0^{(p)})^2 & \mbox{for } m = p+1,
		\end{cases}
	\end{align*}
	with the general solution for $\sigma_{p,m}^2$, 
	\begin{align*}
		 \begin{cases}
			\medskip (\phi_0^{(p)})^ 2 & \mbox{for } m = p+1, \\	
			\medskip \displaystyle{\frac{1}{2^2}}\left((\phi_0^{(p)})^ 2 + (\phi_0^{(p)} + \phi_{1}^{(p)})^2 \right) & \mbox{for } m = p+2, \\	
			\hspace*{2.5cm}\vdots & \hspace*{1cm}\vdots \\
			\medskip \displaystyle{\frac{1}{\ell^2}}\left( (\phi_0^{(p)})^ 2 + (\phi_0^{(p)} + \phi_{1}^{(p)})^2 + \cdots + (\phi_0^{(p)} + \cdots + \phi_{\ell-1}^{(p)})^2 \right) & \mbox{for } m = p +\ell, \\	
			\hspace*{2.5cm}\vdots & \hspace*{1cm}\vdots \\
			\medskip \displaystyle{\frac{1}{p^2}}\left( (\phi_0^{(p)})^ 2 + (\phi_0^{(p)} + \phi_{1}^{(p)})^2 + \cdots + (\phi_0^{(p)} + \cdots + \phi_{p-1}^{(p)})^2 \right) & \mbox{for } m = 2p, \\	
			\medskip \displaystyle{\frac{1}{(p+k)^2}}\left( p^2 \sigma_{p,2p}^2 + k(\phi_0^{(p)} + \cdots + \phi_{p}^{(p)})^2 \right) & \mbox{for } m = 2p+k,\, k=1,2,\ldots. \\	
		\end{cases}
	\end{align*}
\end{theorem}

\subsection{\texttt{Rollage} Algorithm for Fitting \texttt{AR} Models}
\label{sec:rollage_alg}

Based on theoretical results developed in \cref{sec:RA}, we introduce the \texttt{Rollage} algorithm, depicted in \cref{alg:rollage} to fit an appropriate \texttt{AR} model to a given time series data. 
\begin{algorithm}[!h]
	\caption{\texttt{Rollage}: A Novel Algorithm for \texttt{AR} Fitting}	
	\begin{algorithmic}
		\STATE \textbf{Input:} 
		\begin{itemize}[label=-]
			\vspace{2mm}
			\item Time series data $\{y_1,\ldots,y_n\}$\,;
			\vspace{2mm}
			\item A relatively large value $\bar{p} \ll n$\,;
		\end{itemize}
		\vspace{2mm}
		\item \emph{Step 0}. Set $\ell=1$\,;
		\vspace{2mm}
		\WHILE {$ \ell < \bar{p} $} 
			\vspace{2mm}
			\STATE \emph{Step 1}. $\ell \leftarrow \ell + 1$\,;
			\vspace{2mm}
			\STATE \emph{Step 2}. Find the \textsf{CMLE} of the coefficient vector, $(\hat{\phi}_{1}^{(\ell)},\ldots,\hat{\phi}_\ell^{(\ell)})$\,;
			\vspace{2mm}
			\STATE \emph{Step 3}. Compute the values of rolling averages, $\bar{\phi}_{h,\ell}$, for $h=1,\ldots,\ell-1$ as in \cref{def:rolling_average}\,;
			\vspace{2mm}
			\STATE \emph{Step 4}. Compute the variance of rolling averages, $\hat{\sigma}_{\ell,m}^2$, for $m=\ell+1,\ldots,\bar{p}$ as in \cref{thm:lim_dist_rolling_average}\,;
			\vspace{2mm}
		\ENDWHILE
		\vspace{2mm}
		\STATE \emph{Step 5}. Estimate $p$ as the largest $\ell$ such that at least $5\%$ of inequalities $ |\bar{\phi}_{\ell,m}| \geq 1.96\sigma_{\ell,m}/\sqrt{n-\bar{p}} $ for $m = \ell+1,\ldots,\bar{p}$ are held\,;
		\vspace{2mm}
		\STATE \textbf{Output:} Estimated order $p$ and coefficients $(\hat{\phi}_{1}^{(p)},\ldots,\hat{\phi}_p^{(p)})$.
	\end{algorithmic}
	\label{alg:rollage}
\end{algorithm}

\paragraph{Typical procedure for calculating rolling average. } \cref{tab:MLE_phi,tab:rolling_average} illustrate how the rolling averages are calculated in the \texttt{Rollage} algorithm.
\begin{table}[h!]
	\begin{center}
		\begin{tabular}{ | c | l |}
			\hline
			\textbf{Model} & \textbf{Parameters}  \\ \hline
			$\mathtt{AR}(1)$ & $\hat{\phi}_1^{(1)}$  \\ \hline
			$\mathtt{AR}(2)$ & $\hat{\phi}_1^{(2)}, \hat{\phi}_2^{(2)}$  \\ \hline 
			$\mathtt{AR}(3)$ & $\hat{\phi}_1^{(3)}, \hat{\phi}_2^{(3)}, \hat{\phi}_3^{(3)}$  \\ \hline 
			\vdots & \vdots \\ \hline
			$\mathtt{AR}(\bar{p})$ & $\hat{\phi}_1^{(\bar{p})},\cdots, \hat{\phi}_{\bar{p}}^{(\bar{p})}$  \\ 
			\hline
		\end{tabular}
		\caption{Fitting $\bar{p}$ models $\mathtt{AR}(1), \ldots, \mathtt{AR}(\bar{p})$ to the data and finding the \textsf{MLE} of their parameters}
		\label{tab:MLE_phi}
	\end{center}
\end{table}

\begin{table}[h!]
	\begin{center}
		\begin{tabular}{ | c | l |}
			\hline
			\textbf{Model} & \textbf{Rolling average} \\ \hline
			$\mathtt{AR}(1)$ & $-$  \\ \hline
			$\mathtt{AR}(2)$ & $\bar{\phi}_{1,2}=\hat{\phi}_2^{(2)}$  \\ \hline 
			$\mathtt{AR}(3)$ & $\bar{\phi}_{1,3}=\frac{1}{2}(\hat{\phi}_2^{(3)} + \hat{\phi}_3^{(3)})$, $\bar{\phi}_{2,3}=\hat{\phi}_3^{(3)}$  \\ \hline 
			\vdots & \vdots \\ \hline
			$\mathtt{AR}(\bar{p})$ & $\bar{\phi}_{1,\bar{p}}=\frac{1}{\bar{p}-1}(\hat{\phi}_2^{(\bar{p})} + \cdots + \hat{\phi}_{\bar{p}}^{(\bar{p})})$, $\ldots$,  $\bar{\phi}_{\bar{p}-1,\bar{p}}=\hat{\phi}_{\bar{p}}^{(\bar{p})}$  \\
			\hline
		\end{tabular}
		\caption{Calculating the rolling averages}
		\label{tab:rolling_average}
	\end{center}
\end{table}

%
%

\subsection{\texttt{Rollage} Algorithm for Fitting \texttt{ARMA} Models}
\label{sec:rollage_alg_ARMA}

In order to fit \texttt{ARMA} (and \texttt{MA}) models using the theory developed in \cref{sec:rollage_alg}, we must use the \texttt{Rollage} algorithm in conjunction with Durbin's \cref{alg:hr}. To achieve this, we must first introduce a threshold hyper-parameter to \texttt{Rollage} so that an appropriately large $\tilde{p}$ value is chosen for the \texttt{AR} model, in the intermediate step of \cref{alg:hr}. In this scenario, \texttt{Rollage} will act as a stopping criterion to choose an appropriately large $\tilde{p}$ value. 

We define the threshold hyper-parameter, denoted $\delta$, as the following inequality:
\begin{equation}
    \centering
    \delta < \max\left(\frac{|\bar{\phi}_{\ell,m}|}{1.96\sigma_{\ell,m}/\sqrt{n-\bar{p}}}\right),
\end{equation}
whereby the estimate of the large $\tilde{p}$ value for the intermediate \texttt{AR} model is the first value of $l$ such that the inequality is violated. This extension to Durbin's \cref{alg:hr} can be seen in \cref{alg:rollage_ARMA}, and is called \texttt{Rollage*}.

\begin{algorithm}[!h]
	\caption{\texttt{Rollage*}: An Algorithm for \texttt{ARMA} (and \texttt{MA}) Fitting}	
	\begin{algorithmic}
		\STATE \textbf{Input:} 
		\begin{itemize}[label=-]
			\vspace{2mm}
			\item Time series data $\{y_1,\ldots,y_n\}$\,;
			\vspace{2mm}
			\item The estimated order $q$\,;
			\vspace{2mm}
			\item Threshold parameter $\delta$
		\end{itemize}
		\vspace{2mm}
		\item \emph{Step 1}. Using \texttt{Rollage} (\ref{alg:rollage}) choose the sufficiently high order $\tilde{p} > q$ to be the first value of $l$ such that the inequality $\delta < \max\left(\frac{|\bar{\phi}_{\ell,m}|}{1.96\sigma_{\ell,m}/\sqrt{n-\bar{p}}}\right)$ is violated in the \texttt{Rollage} algorithm;
		\vspace{2mm}
		\item \emph{Step 2}. Find the \textsf{CMLE} of the coefficient vector, $(\hat{\pi}_{1},\ldots,\hat{\pi}_{\tilde{p}})$ and compute the white noise approximations $\widetilde{w}_t$ as in \cref{eq:tilda_wt};
		\vspace{2mm}
		\item \emph{Step 3}. Regress the primary time series over q lagged values of $\widetilde{w}_t$ to estimate the coefficient vector, $(\tilde{\theta}_{1},\ldots,\tilde{\theta}_{q})$ as in \cref{eq:regress_tilde_wt};
		\vspace{2mm}

		\STATE \textbf{Output:} Estimated coefficients $(\hat{\tilde{\theta}}_1,\ldots,\hat{\tilde{\theta}}_q)$.
	\end{algorithmic}
	\label{alg:rollage_ARMA}
\end{algorithm}

	
	\section{Empirical Results} \label{sec:empirical_results}

In this section, we present the performance of the \texttt{Rollage} algorithm on several synthetic time series data. The data were simulated using models with randomly generated sets of coefficient parameters. The variance of the error parameter was 1 for each case. The resulting models were confirmed to be causal and invertible using MATLAB's \textit{arima$()$} function. Parameters were randomly generated for $\mathtt{AR}(p)$ and $\mathtt{MA}(q)$ models with $p=5,10,\ldots,100$ and $q=5,10,\ldots,100$. All possible parameter combinations were then utilised to produce $400$ $\mathtt{ARMA}(p,q)$ models.

Numerical analysis is presented in three subsequent sections. In \cref{sec:emp_AR}, the \texttt{Rollage} algorithm is applied to the $\mathtt{AR}$ synthetic data where its overall efficacy is analysed. In \cref{sec:emp_MA}, the \texttt{Rollage} algorithm is applied to the $\mathtt{MA}$ data where it is used as a criterion in Durbin's \cref{alg:hr} and compared to the BIC and GIC criteria for estimating an optimal $\tilde{p}$ value then subsequently fitting an $\mathtt{MA}$ model. The relative errors of parameter estimation are also analysed amongst the \texttt{Rollage}, BIC and GIC criteria. Analogously, \cref{sec:emp_ARMA} compares the \texttt{Rollage}, BIC and GIC criteria in estimating $\mathtt{ARMA}$ models. 

\subsection{Autoregressive Models}
\label{sec:emp_AR}

From each of the 20 $\mathtt{AR}(p)$ models, 500,000 synthetic time series realizations were generated. The \texttt{Rollage} algorithm was applied to each with results being the average of 20 replications. 

For each of the 20 models, the \texttt{Rollage} algorithm correctly identifies the order $p$, showing its efficacy in estimating this parameter. \cref{fig:rolling_average} displays the rolling average graphs produced by the \texttt{Rollage} algorithm at various lags for the data generated from the $\mathtt{AR}(20)$ model. In general, we notice that for lags less than or equal to the true order $p$, at least one rolling average lies significantly outside the $95\%$ confidence boundary. However, for all subsequent lags greater than the true order $p$ all rolling averages lie within the $95\%$ confidence boundary. Specifically, \cref{fig:rolling_average}(a-e) displays lags less than or equal to $p=$ 20, where it can be seen that the rolling averages lie outside the $95\%$ confidence bounds. \cref{fig:rolling_average}(f-i) then shows lags greater than $p=$ 20 where the rolling averages lie within the $95\%$ confidence bounds. This feature is dictated by the inequalities in Step $5$ of the \texttt{Rollage} algorithm. It should be noted that the confidence boundaries in \cref{fig:rolling_average}(a-f) display similar curvature to those in \cref{fig:rolling_average}(g-i), also dictated by Step $5$, however can not be distinguished due to the larger scale in their respective graphs. In this sense, the \texttt{Rollage} graphs can be used as a tool, analogous to the \textsf{PACF}, in choosing the order of an $\mathtt{AR}$ model.

\begin{figure}[h!]
	\centering
	\begin{subfigure}{.3\textwidth}
		\includegraphics[width=\textwidth]{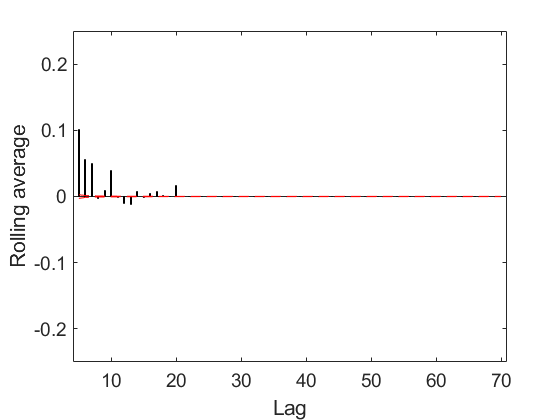}
		\caption{$p=5$}
	\end{subfigure}
	\begin{subfigure}{.3\textwidth}
		\includegraphics[width=\textwidth]{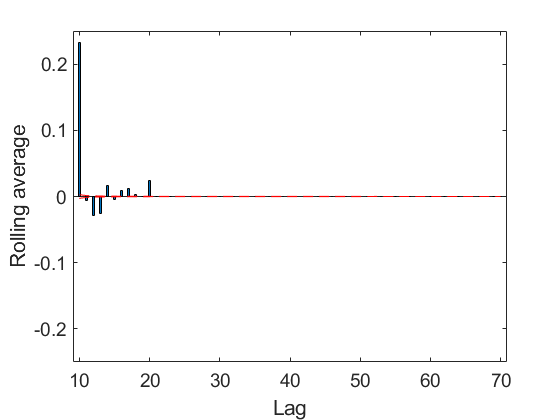}
		\caption{$p=10$ }
	\end{subfigure}
	\begin{subfigure}{.3\textwidth}
		\includegraphics[width=\textwidth]{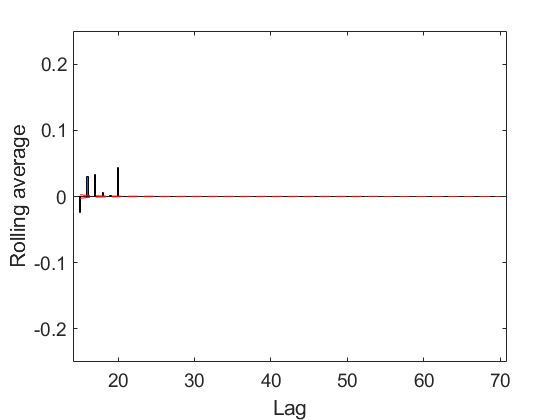}
		\caption{$p=15$}
	\end{subfigure}
	\begin{subfigure}{.3\textwidth}
		\includegraphics[width=\textwidth]{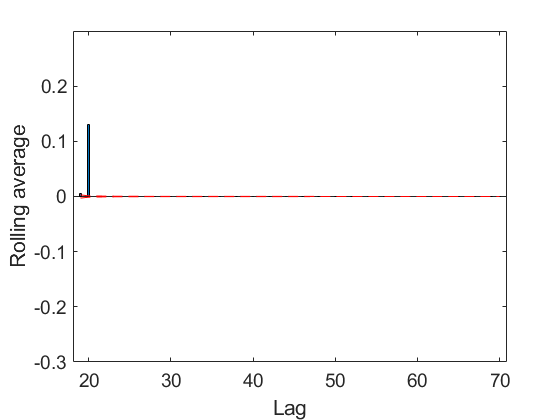}
		\caption{$p=19$}
	\end{subfigure}
	\begin{subfigure}{.3\textwidth}
		\includegraphics[width=\textwidth]{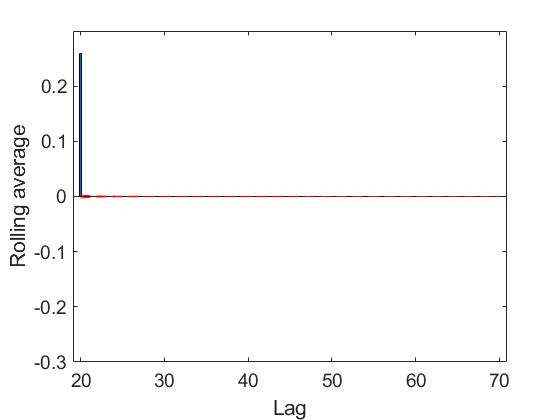}
		\caption{$p=20$}
	\end{subfigure}
	\begin{subfigure}{.3\textwidth}
		\includegraphics[width=\textwidth]{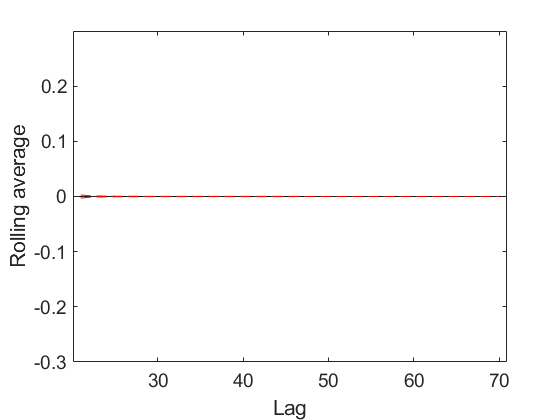}
		\caption{$p=21$}
	\end{subfigure}
	\begin{subfigure}{.3\textwidth}
		\includegraphics[width=\textwidth]{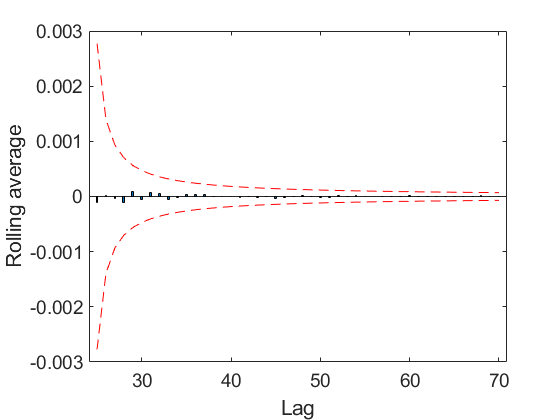}
		\caption{$p=25$}
	\end{subfigure}
	\begin{subfigure}{.3\textwidth}
		\includegraphics[width=\textwidth]{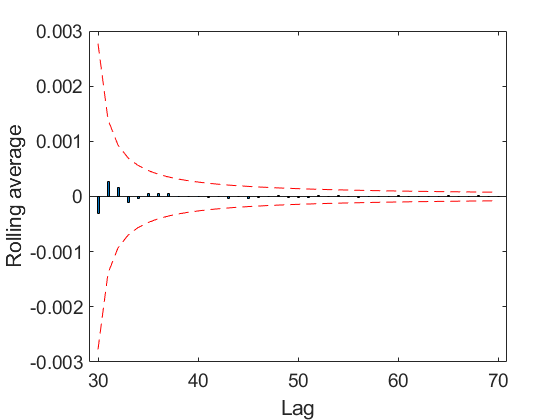}
		\caption{$p=30$}
	\end{subfigure}
	\begin{subfigure}{.3\textwidth}
		\includegraphics[width=\textwidth]{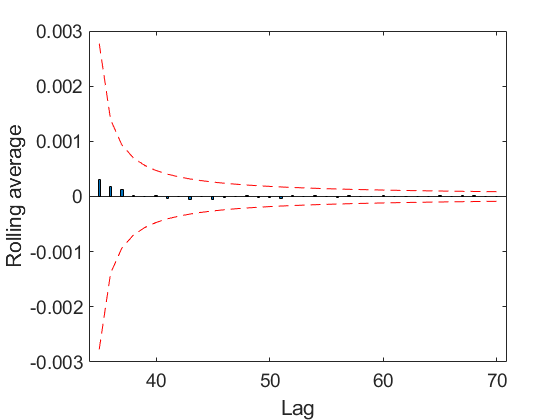}
		\caption{$p=35$}
	\end{subfigure}
	\caption{This figure illustrates the rolling averages at different lags for nine $\mathtt{AR}(p)$ models where the underlying data is generated from an $\mathtt{AR}(20)$ model. It is readily seen that up to $p=$ 20, there exists at least one rolling average which is significantly out of the $95\%$ confidence boundaries. However, from $p=$ 21 onwards, all rolling averages at all lags lie within the $95\%$ confidence bounds.}
	\label{fig:rolling_average}
\end{figure}

\subsection{Moving Average Models}
\label{sec:emp_MA}


For each of the 20 $\mathtt{MA}(q)$ models, synthetic time series data was generated for sample sizes $n=$ 10,000, 20,000, 50,000, 100,000, 200,000, 500,000 and 1,000,000. For each dataset, the \texttt{Rollage}, BIC and GIC criteria are used to estimate an optimal $\tilde{p}$ value for fitting an $\mathtt{AR}$ model before subsequently fitting an $\mathtt{MA}$ model. The threshold hyper-parameter for the \texttt{Rollage} algorithm was varied from 2.5 to 4 in 0.25 increments throughout the experiments with threshold equal to 3 providing the best trade off between estimating an optimal $\tilde{p}$ and producing the lowest relative error of parameter estimates, on average. As a result, all subsequent empirical results utilise a threshold parameter of 3 in the \texttt{Rollage} algorithm. 

\cref{fig:MA_ptilde} displays the estimated $\tilde{p}$ values for each $q$ (i.e. $\tilde{p}$ vs $q$) from the 20 $\mathtt{MA}$ models at the various sample sizes $n$ (excluding $n=$ 10,000). More precisely, the blue, red and magenta graphs are associated with the \texttt{Rollage}, BIC and GIC criteria, respectively. It can be observed that for all sample sizes, the \texttt{Rollage} algorithm produces a smaller than or equal to $\tilde{p}$ value compared to both the BIC and GIC criteria, with this becoming more noticeable as $q$ and $n$ increases. \cref{tab:MA_RE} highlights this, comparing the average $\tilde{p}$ and associated relative error produced by the \texttt{Rollage}, BIC and GIC criteria for each of the respective sample sizes $n$. On average, the \texttt{Rollage} algorithm provides smaller estimates for $\tilde{p}$ than both the BIC and GIC criteria, with the difference between \texttt{Rollage} and the other criteria increasing as $n$ increases. The bottom 2 rows of \cref{tab:MA_RE} quantify this relative difference between \texttt{Rollage} and the BIC and GIC criteria. The relative difference of $\tilde{p}$ is calculated as
\begin{equation}
        \centering
        \frac{|\tilde{p}_{\texttt{A}} - \tilde{p}_{\texttt{R}}|} {\tilde{p}_{\texttt{R}}}
\label{eqn:RelDiff}
\end{equation}
where $\tilde{p}_{\texttt{R}}$ is given by the \texttt{Rollage} algorithm and $\tilde{p}_{\texttt{A}}$ is suggested by the mentioned alternative. The total average column of \cref{tab:MA_RE} summarises this trend for all sample sizes, showing that the \texttt{Rollage} criterion provides $\tilde{p}$ values that are $18.06\%$ smaller than the BIC and $32.26\%$ smaller than the GIC criteria, on average. The results of \cref{fig:MA_ptilde} and \cref{tab:MA_RE} provide empirical evidence that the \texttt{Rollage} algorithm is computationally less expensive than both the BIC and GIC criteria, on average. 

\begin{figure}[t!]
	\centering
	\begin{subfigure}{.3\textwidth}
		\includegraphics[width=\textwidth]{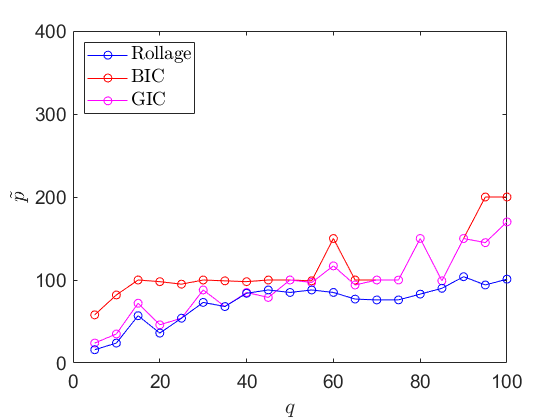}
		\caption{$n=20,000$}
	\end{subfigure}
	\begin{subfigure}{.3\textwidth}
		\includegraphics[width=\textwidth]{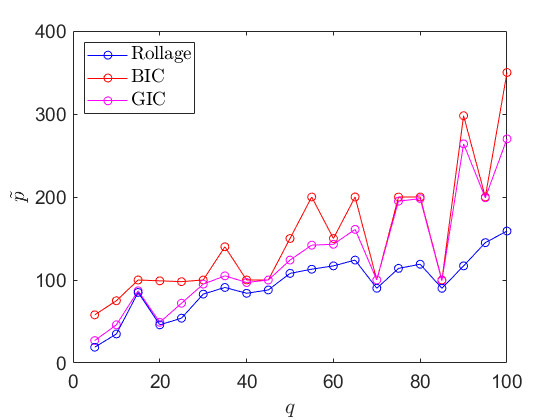}
		\caption{$n=50,000$}
	\end{subfigure}
	\begin{subfigure}{.3\textwidth}
		\includegraphics[width=\textwidth]{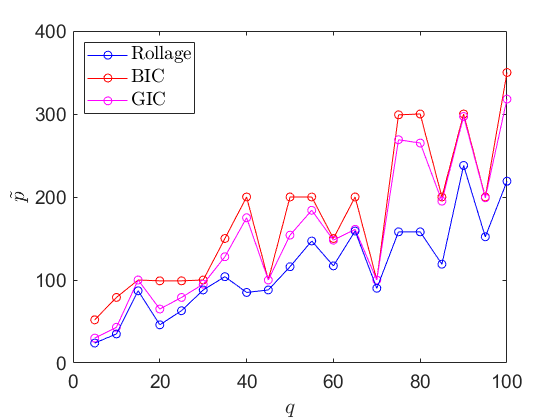}
		\caption{$n=100,000$}
	\end{subfigure}
	\begin{subfigure}{.3\textwidth}
		\includegraphics[width=\textwidth]{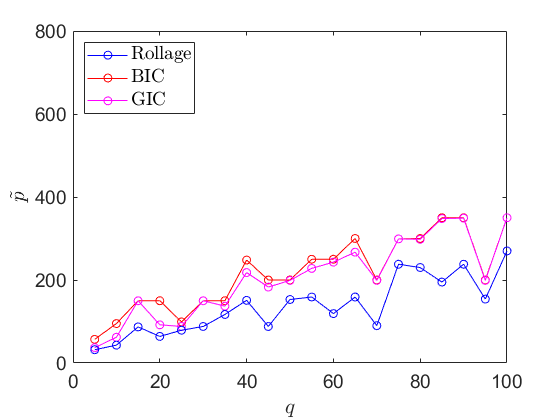}
		\caption{$n=200,000$}
	\end{subfigure}
	\begin{subfigure}{.3\textwidth}
		\includegraphics[width=\textwidth]{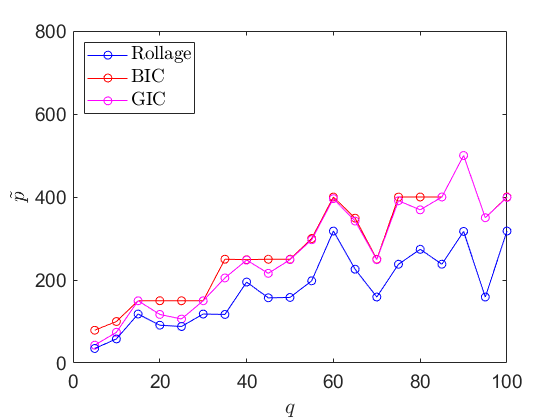}
		\caption{$n=500,000$}
	\end{subfigure}
	\begin{subfigure}{.3\textwidth}
		\includegraphics[width=\textwidth]{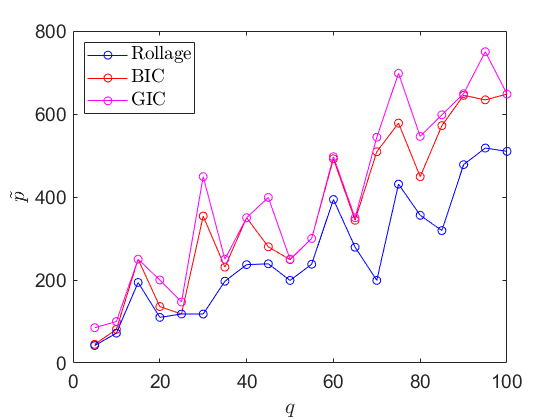}
		\caption{$n=1,000,000$}
	\end{subfigure}
	\caption{Each graph illustrates the $\tilde{p}$ values estimated by the \texttt{Rollage}, BIC and GIC criteria to initially fit an $\mathtt{AR}(p)$ model and then subsequently fit an associated $\mathtt{MA}(q)$ model. The \texttt{Rollage}, BIC and GIC criteria are represented by the blue, red and magenta graphs respectively. The six graphs represent the different sample sizes $n$ that were used to fit an $\mathtt{MA}(q)$ model. It is readily seen that the \texttt{Rollage} criteria regularly produces the smallest estimate for $\tilde{p}$.}
	\label{fig:MA_ptilde}
\end{figure}

Similarly, \cref{fig:ptildevn} displays the estimated $\tilde{p}$ values for each sample size $n$ (i.e. $\tilde{p}$ vs $n$) for various $\mathtt{MA}(q)$ models where a similar conclusion, that the \texttt{Rollage} criterion provides smaller estimates for $\tilde{p}$ than the BIC and GIC criteria, is reached. Additionally, \cref{fig:MA_ptilde} suggests that $\tilde{p}$ increases linearly as $q$ increases, and becomes more evident as $n$ increases, whereas \cref{fig:ptildevn} suggests that $\tilde{p}$ increases logarithmically as n increases. We have fitted linear models to predict the optimal $\tilde{p}$, which find $q$, $\log(n)$ and $q\log(n)$ to be statistically significant predictors for $\tilde{p}$ across all three criteria. \cref{tab:linear_models_MA} provides the coefficients for linear models to predict the optimal $\tilde{p}$ for each of the \texttt{Rollage}, BIC and GIC criteria. For the \texttt{Rollage} algorithm, the rearranged linear model is
\begin{equation}
        \centering
        \hat{\tilde{p}} = q(0.81\log(n)-7.12)+3.19\log(n)
\end{equation}
highlighting that $\tilde{p}$ is heavily dependent on $q$, however the sample size $n$ is less impactful as $\log(n)$ stunts this impact. The linear models can't be applied generally to prescribe a $\tilde{p}$, even within the range of $q$ and $n$ considered, as the set of 20 $\mathtt{MA}$ model coefficients may not be representative of the entire population. The models do however provide a good initial approximation for $\tilde{p}$ and offer insight into variable relationships.  

Analogous to \cref{fig:MA_ptilde}, \cref{fig:MA_ra} shows the corresponding relative error associated with estimating the parameters of an $\mathtt{MA}(q)$ model, in percentage. Here we define the relative error of parameter estimates by 
\begin{equation}
        \centering
        \frac{||\boldsymbol{\hat{\theta}}-\boldsymbol{\theta}||}{||\boldsymbol{\theta}||}
\end{equation}
where $\boldsymbol{\hat{\theta}}$ is the parameter estimates and $\boldsymbol{\theta}$ is the true parameter values for an $\mathtt{MA}(q)$ model. It is observed that for all sample sizes, the \texttt{Rollage} criterion produces similar relative errors to the BIC and GIC criteria over the whole range of $\mathtt{MA}(q)$ models. \cref{tab:MA_RE} summarises this by showing the average relative errors below the average estimate of $\tilde{p}$. For all criteria, the relative error of parameter estimates becomes smaller as $n$ increases and for big data regimes (i.e. $n=$ 1,000,000) all criteria produce relative errors of approximately $1.5\%$. As a result, the \texttt{Rollage} algorithm provides a great trade off between computational runtime and algorithmic accuracy for fitting an $\mathtt{MA}$ model, when compared to the BIC and GIC criteria.

In addition, \cref{fig:MA_ra} suggests that the relative error increases linearly as a function of $q$, however as $n$ increases the linear relationship appears to become more constant across all values of $q$ (particularly $n=$ 1,000,000). In contrast, \cref{fig:revn} displays the relative error of parameter estimates as a function of the sample size $n$ for various $\mathtt{MA}(q)$ models, where we see the relative errors decrease exponentially. This is further evidence that the \texttt{Rollage} algorithm may be appropriate for fitting $\mathtt{MA}$ models to big time series data when compared to current
alternatives.

\begin{table}[]
\resizebox{\textwidth}{!}{\begin{tabular}{@{}ccccccccc@{}}
\toprule
\rowcolor[HTML]{EFEFEF} 
 & \textbf{10k} & \textbf{20k} & \textbf{50k} & \textbf{100k} & \textbf{200k} & \textbf{500k} & \textbf{1M} & \textbf{Total Average} \\ \midrule
\cellcolor[HTML]{EFEFEF}\textbf{\texttt{Rollage}} & \begin{tabular}[c]{@{}c@{}}71\\ (14.09$\%$)\end{tabular} & \begin{tabular}[c]{@{}c@{}}86\\ (10.26$\%$)\end{tabular} & \begin{tabular}[c]{@{}c@{}}115\\ (7.14$\%$)\end{tabular} & \begin{tabular}[c]{@{}c@{}}134\\ (5.76$\%$)\end{tabular} & \begin{tabular}[c]{@{}c@{}}168\\ (3.93$\%$)\end{tabular} & \begin{tabular}[c]{@{}c@{}}216\\ (2.99$\%$)\end{tabular} & \begin{tabular}[c]{@{}c@{}}297\\ (1.78$\%$)\end{tabular} & \begin{tabular}[c]{@{}c@{}}155\\ (6.56$\%$)\end{tabular} \\ \midrule
\cellcolor[HTML]{EFEFEF}\textbf{BIC} & 
\begin{tabular}[c]{@{}c@{}}76\\ (13.69$\%$)\end{tabular} & \begin{tabular}[c]{@{}c@{}}94\\ (9.75$\%$)\end{tabular} & \begin{tabular}[c]{@{}c@{}}129\\ (6.66$\%$)\end{tabular} & \begin{tabular}[c]{@{}c@{}}155\\ (5.22$\%$)\end{tabular} & \begin{tabular}[c]{@{}c@{}}205\\ (3.47$\%$)\end{tabular} & \begin{tabular}[c]{@{}c@{}}263\\ (2.55$\%$)\end{tabular} & \begin{tabular}[c]{@{}c@{}}363\\ (1.48$\%$)\end{tabular} & 
\begin{tabular}[c]{@{}c@{}}183\\ (6.12$\%$)\end{tabular} \\ \midrule
\cellcolor[HTML]{EFEFEF}\textbf{GIC} & 
\begin{tabular}[c]{@{}c@{}}100\\ (12.79$\%$)\end{tabular} & \begin{tabular}[c]{@{}c@{}}113\\ (9.30$\%$)\end{tabular} & \begin{tabular}[c]{@{}c@{}}151\\ (6.34$\%$)\end{tabular} & \begin{tabular}[c]{@{}c@{}}174\\ (5.07$\%$)\end{tabular} & \begin{tabular}[c]{@{}c@{}}217\\ (3.40$\%$)\end{tabular} & \begin{tabular}[c]{@{}c@{}}276\\ (2.53$\%$)\end{tabular} & \begin{tabular}[c]{@{}c@{}}403\\ (1.42$\%$)\end{tabular} & 
\begin{tabular}[c]{@{}c@{}}205\\ (5.84$\%$)\end{tabular} \\ \midrule
\cellcolor[HTML]{EFEFEF}\textbf{Relative to BIC} & 7.04\% & 9.30\% & 12.17\% & 15.67\% & 22.02\% & 21.76\% & 22.22\% & 18.06\% \\ \midrule
\cellcolor[HTML]{EFEFEF}\textbf{Relative GIC} & 40.85\% & 31.40\% & 31.30\% & 29.85\% & 29.17\% & 27.78\% & 35.69\% & 32.26\%
\end{tabular}}
\caption{Each cell provides the average $\tilde{p}$ (rounded to the nearest integer) and average associated relative error (rounded to 2 d.p.) of the subsequently fitted $\mathtt{MA}$ models for the given sample size and estimation method.}
\label{tab:MA_RE}
\end{table}

\begin{figure}[h!]
	\centering
	\begin{subfigure}{.3\textwidth}
		\includegraphics[width=\textwidth]{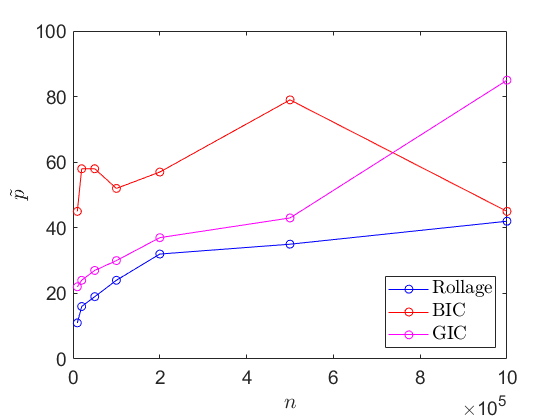}
		\caption{$\mathtt{MA}(5)$}
	\end{subfigure}
	\begin{subfigure}{.3\textwidth}
		\includegraphics[width=\textwidth]{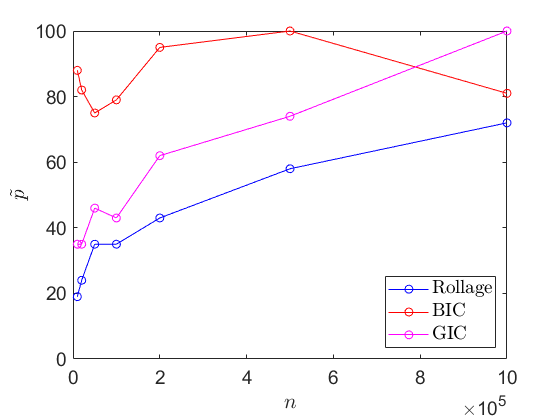}
		\caption{$\mathtt{MA}(10)$}
	\end{subfigure}
	\begin{subfigure}{.3\textwidth}
		\includegraphics[width=\textwidth]{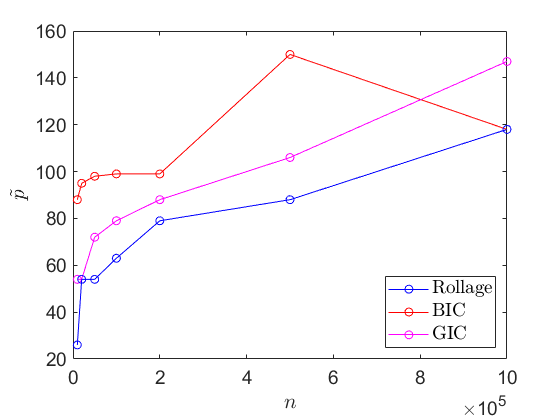}
		\caption{$\mathtt{MA}(25)$}
	\end{subfigure}
	\begin{subfigure}{.3\textwidth}
		\includegraphics[width=\textwidth]{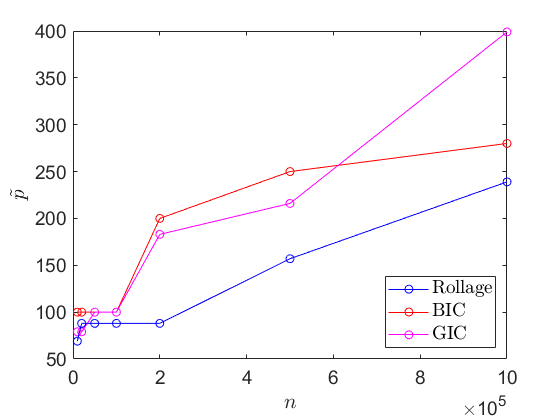}
		\caption{$\mathtt{MA}(45)$}
	\end{subfigure}
	\begin{subfigure}{.3\textwidth}
		\includegraphics[width=\textwidth]{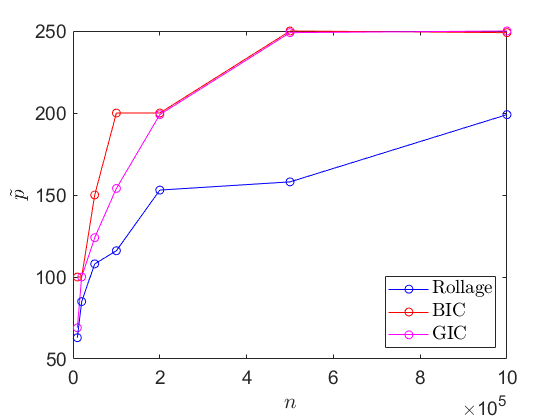}
		\caption{$\mathtt{MA}(50)$}
	\end{subfigure}
	\begin{subfigure}{.3\textwidth}
		\includegraphics[width=\textwidth]{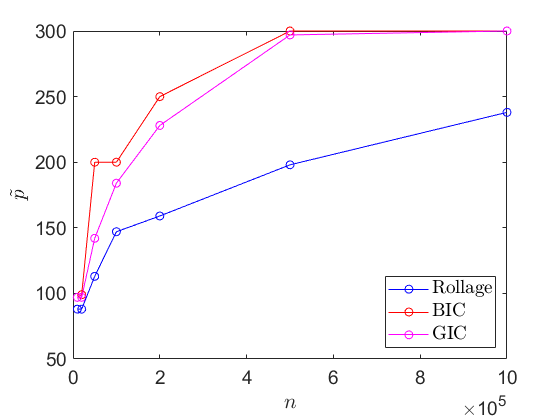}
		\caption{$\mathtt{MA}(55)$}
	\end{subfigure}
	\begin{subfigure}{.3\textwidth}
		\includegraphics[width=\textwidth]{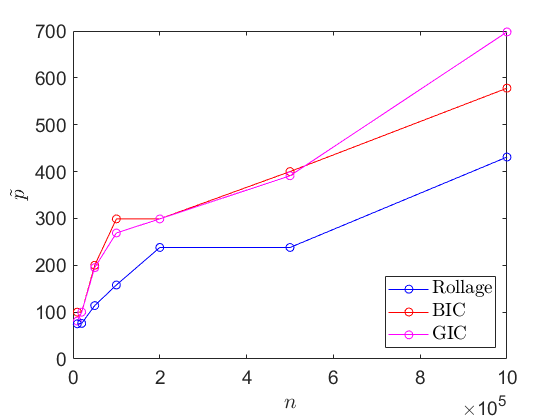}
		\caption{$\mathtt{MA}(75)$}
	\end{subfigure}
	\begin{subfigure}{.3\textwidth}
		\includegraphics[width=\textwidth]{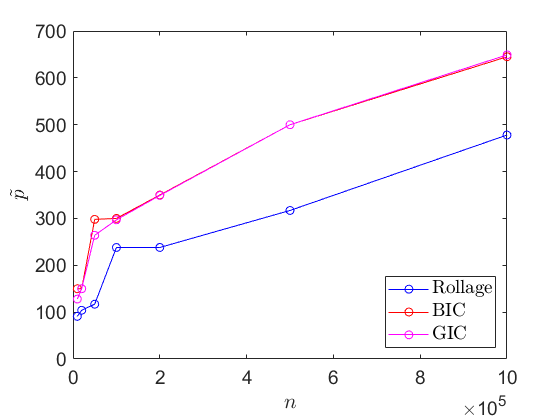}
		\caption{$\mathtt{MA}(90)$}
	\end{subfigure}
	\begin{subfigure}{.3\textwidth}
		\includegraphics[width=\textwidth]{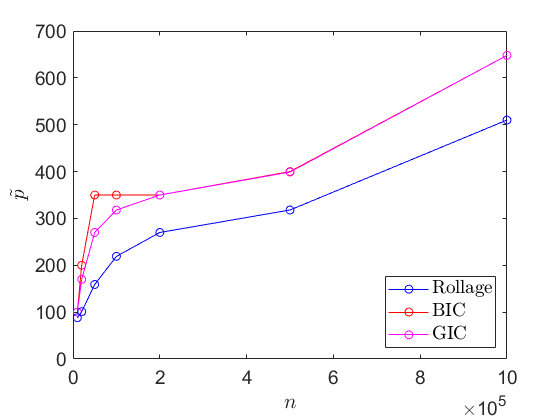}
		\caption{$\mathtt{MA}(100)$}
	\end{subfigure}
	\caption{This figure illustrates $\tilde{p}$ as a function of the sample size $n$ for various $\mathtt{MA}(q)$ models. The \texttt{Rollage} criteria is shown to consistently produce the smallest estimate for $\tilde{p}$. A logarithmic relationship can be observed between $\tilde{p}$ and the sample size $n$.}
	\label{fig:ptildevn}
\end{figure}

\begin{table}[]
\centering
\begin{tabular}{ccccccc}
\hline
\rowcolor[HTML]{EFEFEF} 
 &  \textbf{$q$} & \textbf{$\log(n)$} & \textbf{$q\log(n)$} & \textbf{$R^2_a$} \\ \hline
\cellcolor[HTML]{EFEFEF}\textbf{\texttt{Rollage}} & -7.12 & 3.19 & 0.81 & 0.95 \\ \hline
\cellcolor[HTML]{EFEFEF}BIC & -9.39 & 3.27 & 1.06 & 0.95 \\ \hline
\cellcolor[HTML]{EFEFEF}GIC & -9.43 & 5.76 & 1.05 & 0.94 \\ \hline
\end{tabular}
\caption{Coefficients for the corresponding terms of linear models fitted to predict the optimal $\tilde{p}$ chosen by each criterion based on all $\mathtt{MA}$ simulated time series data.}
\label{tab:linear_models_MA}
\end{table}

\begin{figure}[h!]
	\centering
	\begin{subfigure}{.3\textwidth}
		\includegraphics[width=\textwidth]{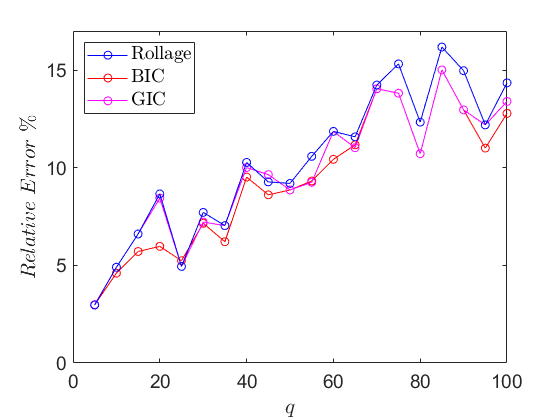}
		\caption{$n=20,000$}
	\end{subfigure}
	\begin{subfigure}{.3\textwidth}
		\includegraphics[width=\textwidth]{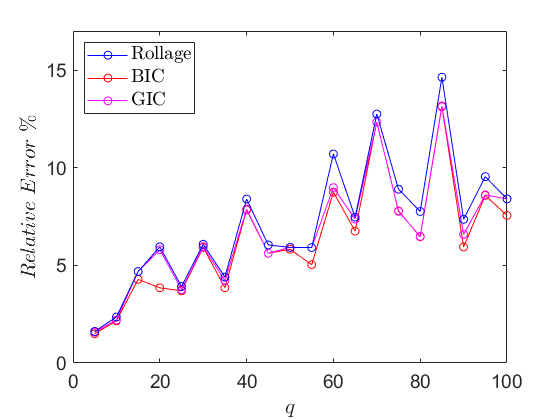}
		\caption{$n=50,000$}
	\end{subfigure}
	\begin{subfigure}{.3\textwidth}
		\includegraphics[width=\textwidth]{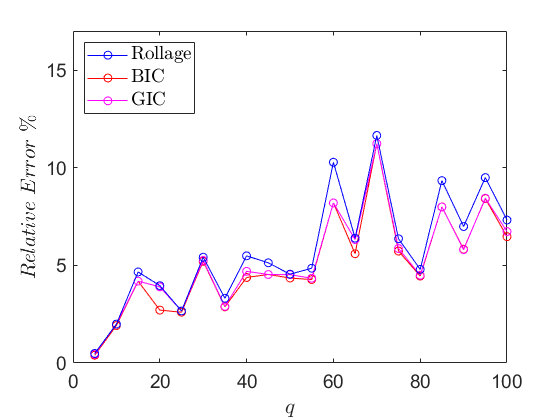}
		\caption{$n=100,000$}
	\end{subfigure}
	\begin{subfigure}{.3\textwidth}
		\includegraphics[width=\textwidth]{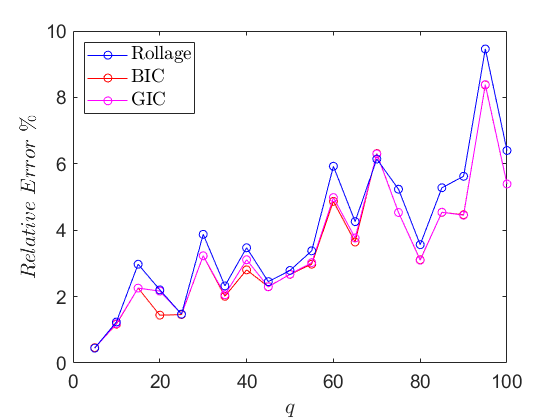}
		\caption{$n=200,000$}
	\end{subfigure}
	\begin{subfigure}{.3\textwidth}
		\includegraphics[width=\textwidth]{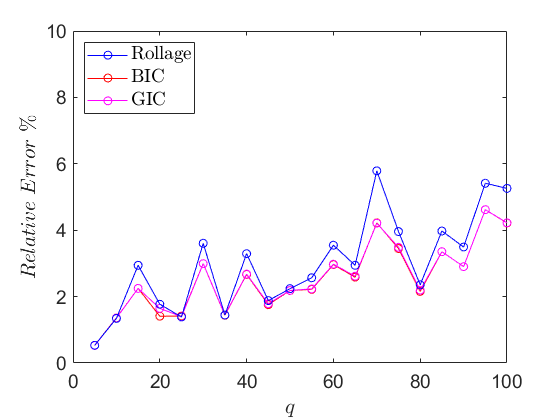}
		\caption{$n=500,000$}
	\end{subfigure}
	\begin{subfigure}{.3\textwidth}
		\includegraphics[width=\textwidth]{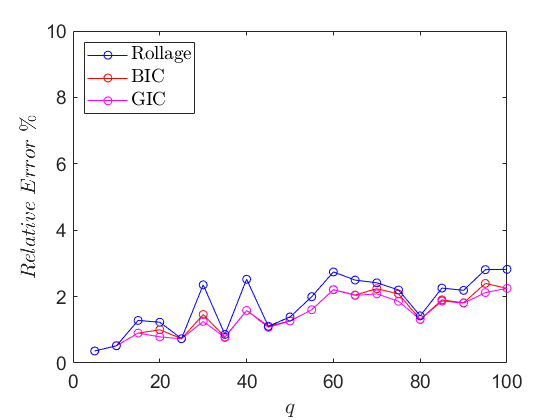}
		\caption{$n=1,000,000$}
	\end{subfigure}
	\caption{Each graph illustrates the $Relative$ $Error$ $\%$ of estimating the parameters of an $\mathtt{MA}(q)$ model for each of the \texttt{Rollage}, BIC and GIC criteria. The \texttt{Rollage}, BIC and GIC criteria are represented by the blue, red and magenta graphs respectively. The six graphs represent the different sample sizes $n$ that were used to fit an $\mathtt{MA}(q)$ model. All three criteria produce similar relative errors in estimating the parameters of an associated $\mathtt{MA}(q)$ model. Relative errors reduce as sample size $n$ increases for all criteria.}
	\label{fig:MA_ra}
\end{figure}

\begin{figure}[h!]
	\centering
	\begin{subfigure}{.3\textwidth}
		\includegraphics[width=\textwidth]{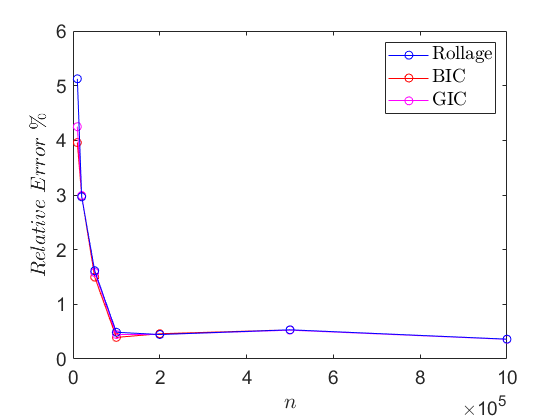}
		\caption{$\mathtt{MA}(5)$}
	\end{subfigure}
	\begin{subfigure}{.3\textwidth}
		\includegraphics[width=\textwidth]{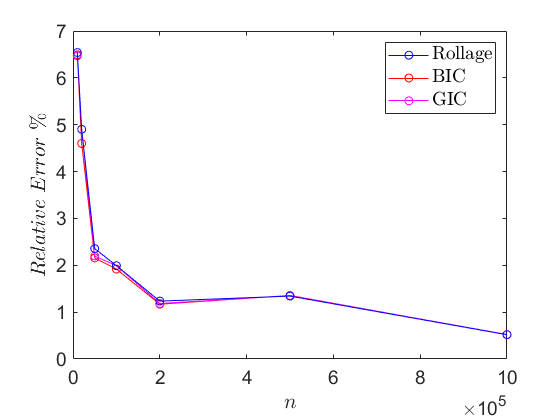}
		\caption{$\mathtt{MA}(10)$}
	\end{subfigure}
	\begin{subfigure}{.3\textwidth}
		\includegraphics[width=\textwidth]{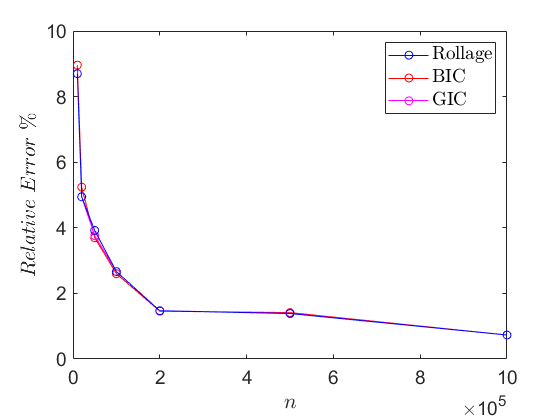}
		\caption{$\mathtt{MA}(25)$}
	\end{subfigure}
	\begin{subfigure}{.3\textwidth}
		\includegraphics[width=\textwidth]{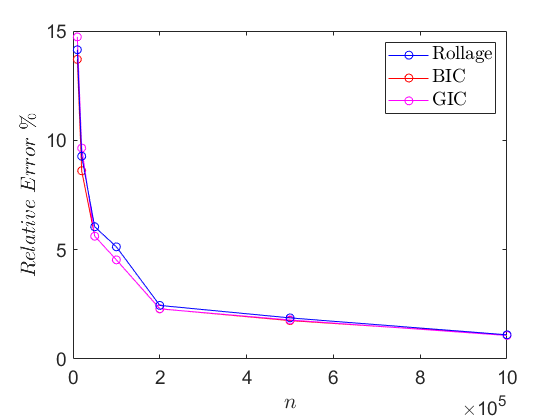}
		\caption{$\mathtt{MA}(45)$}
	\end{subfigure}
	\begin{subfigure}{.3\textwidth}
		\includegraphics[width=\textwidth]{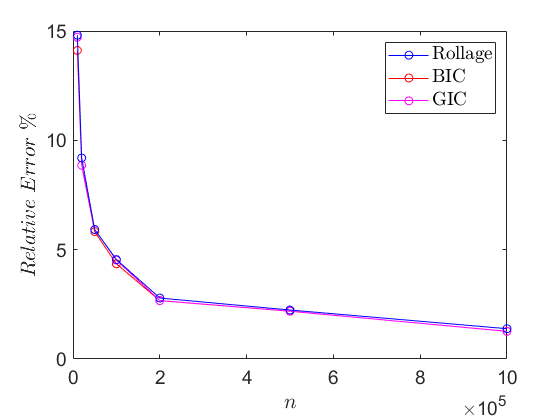}
		\caption{$\mathtt{MA}(50)$}
	\end{subfigure}
	\begin{subfigure}{.3\textwidth}
		\includegraphics[width=\textwidth]{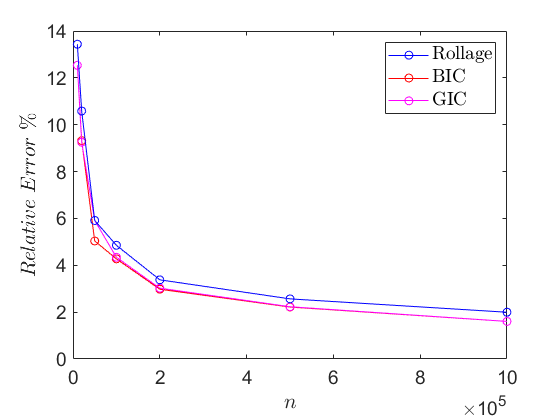}
		\caption{$\mathtt{MA}(55)$}
	\end{subfigure}
	\begin{subfigure}{.3\textwidth}
		\includegraphics[width=\textwidth]{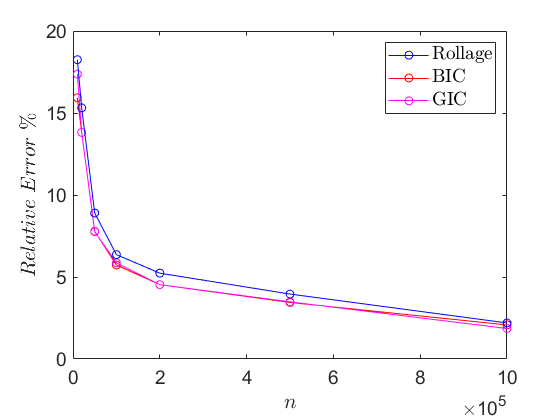}
		\caption{$\mathtt{MA}(75)$}
	\end{subfigure}
	\begin{subfigure}{.3\textwidth}
		\includegraphics[width=\textwidth]{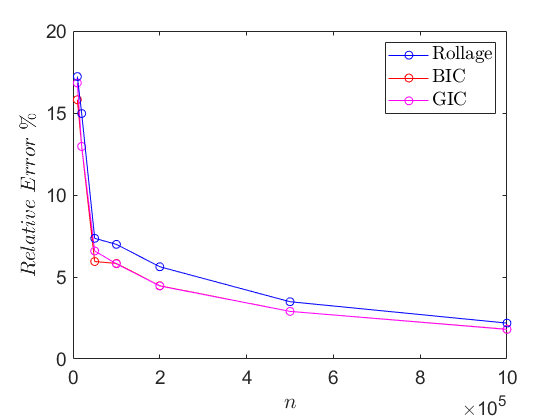}
		\caption{$\mathtt{MA}(90)$}
	\end{subfigure}
	\begin{subfigure}{.3\textwidth}
		\includegraphics[width=\textwidth]{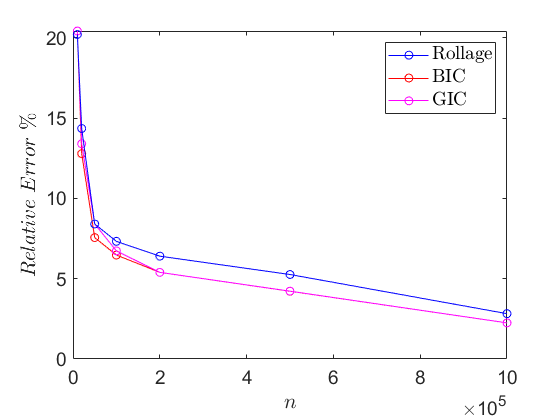}
		\caption{$\mathtt{MA}(100)$}
	\end{subfigure}
	\caption{This figure illustrates $Relative$ $Error$ $\%$ of estimating the parameters of an $\mathtt{MA}(q)$ model as a function of the sample size $n$. The \texttt{Rollage} criterion is shown to produce similar relative errors when compared to the BIC and GIC criteria. A negative exponential relationship can be observed between $Relative$ $Error$ $\%$ and the sample size $n$.}
	\label{fig:revn}
\end{figure}

\subsection{Autoregressive Moving Average Models}
\label{sec:emp_ARMA}
For each of the 400 $\mathtt{ARMA}(p,q)$ models, synthetic time series data was generated for sample sizes $n=$ 10,000, 20,000, 50,000, 100,000, 200,000, 500,000 and 1,000,000. 
For each dataset, the \texttt{Rollage}, BIC and GIC criteria produced estimates of the optimal $\tilde{p}$ value for fitting a large $\mathtt{AR}$ model before subsequently fitting the full $\mathtt{ARMA}$ model. The threshold hyper-parameter for the \texttt{Rollage} algorithm was varied from 2.5 to 4 in 0.25 increments throughout the experiments and 3 was found to provide the lowest average relative errors for the parameter estimates over the various $n<$ 1,000,000 and 3.5 for $n=$ 1,000,000. 

Similar to the case with $\mathtt{MA}$ models, Figure~\ref{fig:pt_vs_n_relerr} provides examples of the \texttt{Rollage} algorithm consistently providing smaller values of $\tilde{p}$ than the other criteria. The total average column of Table~\ref{tab:RE} summarises this trend for all considered sample sizes, showing also that GIC consistently provides the largest values. \texttt{Rollage} begins suggesting significantly smaller $\tilde{p}$ than BIC for $n>$ 50,000, and for all considered sample sizes for GIC. The bottom 2 rows of Table~\ref{tab:RE} quantify this relative difference, which is calculated as per \cref{eqn:RelDiff}.
The relative difference becomes significant as $n$ gets large. However, the \texttt{Rollage} algorithm provides relative errors that closely matches those for BIC and GIC for all considered sample sizes. When $n=$ 1,000,000, we see that BIC provides a relative error only $1.5\%$ lower than \texttt{Rollage}, but the $\tilde{p}$ prescribed by BIC is almost $20\%$ larger than that for \texttt{Rollage}. This is evidence that the \texttt{Rollage} algorithm could provide significant computational savings within the context of big data with a limited reduction in parameter accuracy. 

Figures~\ref{fig:pt_vs_n_relerr} and~\ref{fig:pt_vs_p_q} suggest that $\tilde{p}$ grows linearly as a function of $q$ and logarithmically as a function of $n$. The relationship with $p$ is less obvious, but linear models find $q$, $p$, $p\log(n)$, $q\log(n)$ and $\log(n)$ to be significant predictors for $\tilde{p}$ across all three criterion. The coefficients of the three linear models are provided in Table~\ref{tab:linear_models}. The dependence of $\tilde{p}$ suggested by \texttt{Rollage} on $p$ and $q$ is clarified if we consider the rearranged model
\begin{equation}
        \centering
        \hat{\tilde{p}} = p(2.08-0.14\log(n))+q(0.76\log(n)-6.29)+4.19\log(n).
\end{equation}

Here we can see that for the range of sample sizes considered for this study, the coefficients of $p$ and $q$ are both positive, but as $n$ increases the coefficient of $p$ decreases while the coefficient of $q$ increases. This is more evidence that for larger $n$, the value of $q$ is more critical for determining the optimal $\tilde{p}$ than $p$. This is likely due to $q$ determining the number of white noise columns included in the design matrix. These regression models can't be applied generally to prescribe a $\tilde{p}$, even within the range of $p$, $q$ and $n$ considered, as the 400 $\mathtt{ARMA}$ models may not be representative of the entire population of $\mathtt{ARMA}$ models. The linear model does however provide a good initial approximation for $\tilde{p}$ and offer insight into variable relationships. 

\begin{table}[]
\resizebox{\textwidth}{!}{\begin{tabular}{@{}ccccccccc@{}}
\toprule
\rowcolor[HTML]{EFEFEF} 
 & \textbf{10k} & \textbf{20k} & \textbf{50k} & \textbf{100k} & \textbf{200k} & \textbf{500k} & \textbf{1M} & \textbf{Total Average} \\ \midrule
\cellcolor[HTML]{EFEFEF}\textbf{\texttt{Rollage}} & \begin{tabular}[c]{@{}c@{}}130\\ (35.94$\%$)\end{tabular} & \begin{tabular}[c]{@{}c@{}}142\\ (32.59$\%$)\end{tabular} & \begin{tabular}[c]{@{}c@{}}168\\ (27.33$\%$)\end{tabular} & \begin{tabular}[c]{@{}c@{}}172\\ (25.32$\%$)\end{tabular} & \begin{tabular}[c]{@{}c@{}}205\\ (20.59$\%$)\end{tabular} & \begin{tabular}[c]{@{}c@{}}269\\ (14.99$\%$)\end{tabular} & \begin{tabular}[c]{@{}c@{}}301\\ (11.09$\%$)\end{tabular} & \begin{tabular}[c]{@{}c@{}}195\\ (23.88$\%$)\end{tabular} \\ \midrule
\cellcolor[HTML]{EFEFEF}\textbf{BIC} & 
\begin{tabular}[c]{@{}c@{}}131\\ (35.24$\%$)\end{tabular} & \begin{tabular}[c]{@{}c@{}}143\\ (31.90$\%$)\end{tabular} & \begin{tabular}[c]{@{}c@{}}169\\ (26.97$\%$)\end{tabular} & \begin{tabular}[c]{@{}c@{}}182\\ (24.12$\%$)\end{tabular} & \begin{tabular}[c]{@{}c@{}}222\\ (19.06$\%$)\end{tabular} & \begin{tabular}[c]{@{}c@{}}292\\ (13.51$\%$)\end{tabular} & \begin{tabular}[c]{@{}c@{}}358\\ (9.68$\%$)\end{tabular} & 
\begin{tabular}[c]{@{}c@{}}211\\ (22.84$\%$)\end{tabular} \\ \midrule
\cellcolor[HTML]{EFEFEF}\textbf{GIC} & 
\begin{tabular}[c]{@{}c@{}}170\\ (35.00$\%$)\end{tabular} & \begin{tabular}[c]{@{}c@{}}186\\ (31.53$\%$)\end{tabular} & \begin{tabular}[c]{@{}c@{}}220\\ (26.76$\%$)\end{tabular} & \begin{tabular}[c]{@{}c@{}}219\\ (23.74$\%$)\end{tabular} & \begin{tabular}[c]{@{}c@{}}263\\ (18.86$\%$)\end{tabular} & \begin{tabular}[c]{@{}c@{}}344\\ (13.72$\%$)\end{tabular} & \begin{tabular}[c]{@{}c@{}}425\\ (9.83$\%$)\end{tabular} & 
\begin{tabular}[c]{@{}c@{}}257\\ (22.69$\%$)\end{tabular} \\ \midrule
\cellcolor[HTML]{EFEFEF}\textbf{Relative to BIC} & 0.77\% & 0.70\% & 0.60\% & 5.81\% & 8.29\% & 8.55\% & 18.94\% & 8.21\% \\ \midrule
\cellcolor[HTML]{EFEFEF}\textbf{Relative GIC} & 30.77\% & 30.99\% & 30.95\% & 27.33\% & 28.29\% & 27.88\% & 41.20\% & 31.79\%
\end{tabular}}
\caption{Each cell provides the average $\tilde{p}$ (rounded to the nearest integer) and average associated relative error (rounded to 2 d.p.) of the subsequently fitted $\mathtt{ARMA}$ models for the given sample size and estimation method.}
\label{tab:RE}
\end{table}


\begin{figure}[h!]
	\centering
	\begin{subfigure}{.48\textwidth}
		\includegraphics[width=\textwidth]{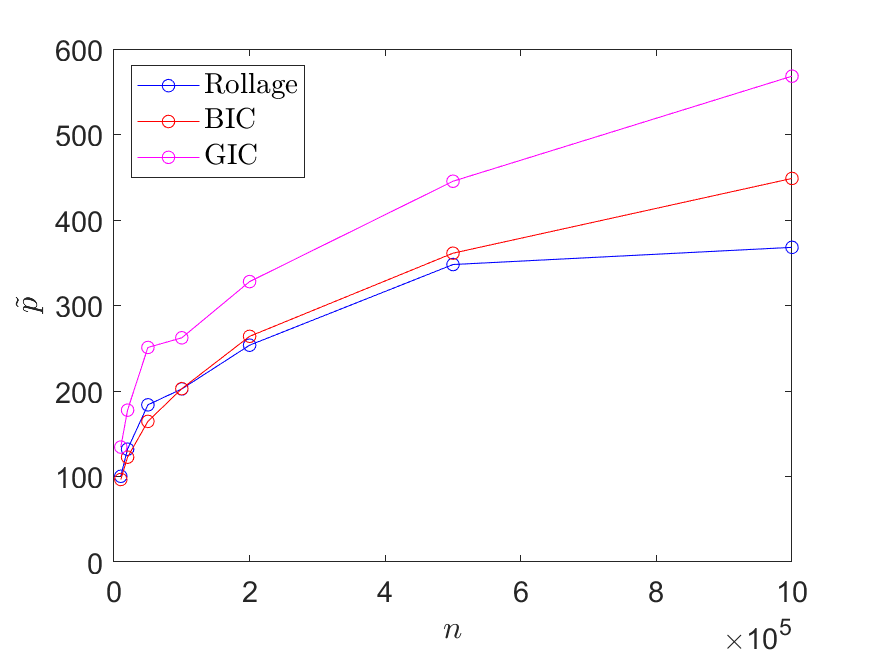}
		\caption{$\mathtt{ARMA}(25,75)$}
	\end{subfigure}
	\begin{subfigure}{.48\textwidth}
		\includegraphics[width=\textwidth]{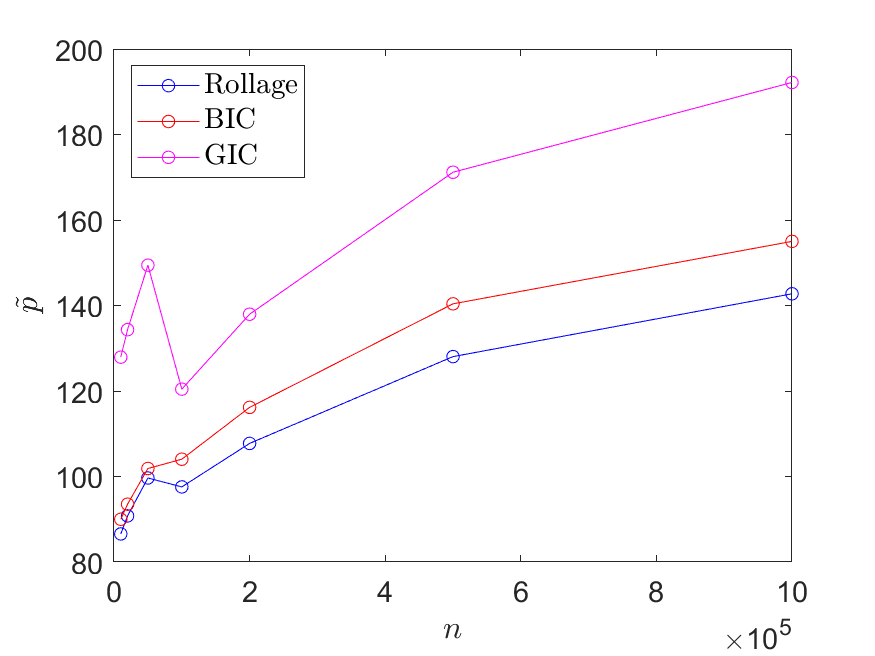}
		\caption{$\mathtt{ARMA}(75,25)$}
	\end{subfigure}
	\begin{subfigure}{.48\textwidth}
		\includegraphics[width=\textwidth]{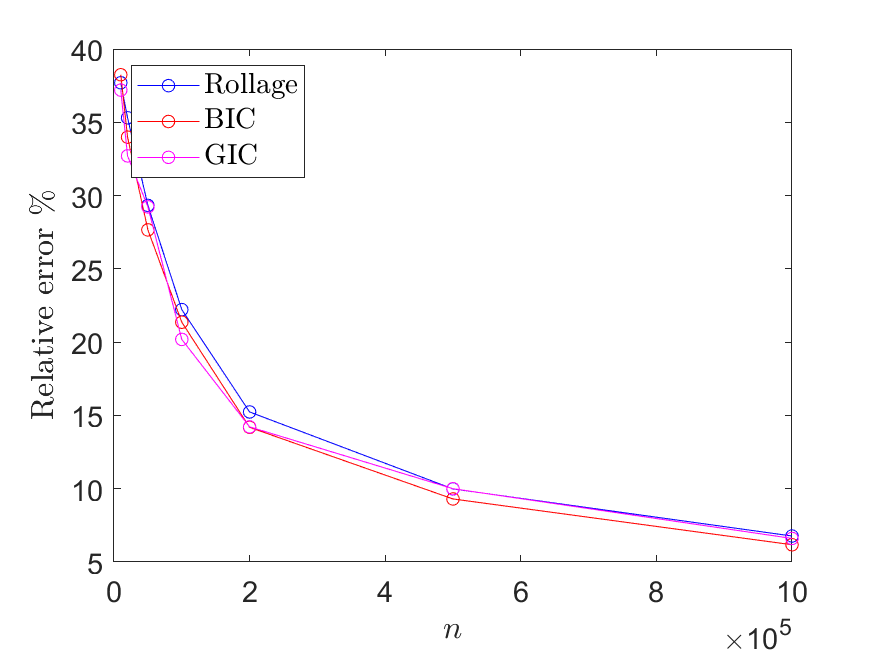}
		\caption{$\mathtt{ARMA}(25,75)$}
	\end{subfigure}
	\begin{subfigure}{.48\textwidth}
		\includegraphics[width=\textwidth]{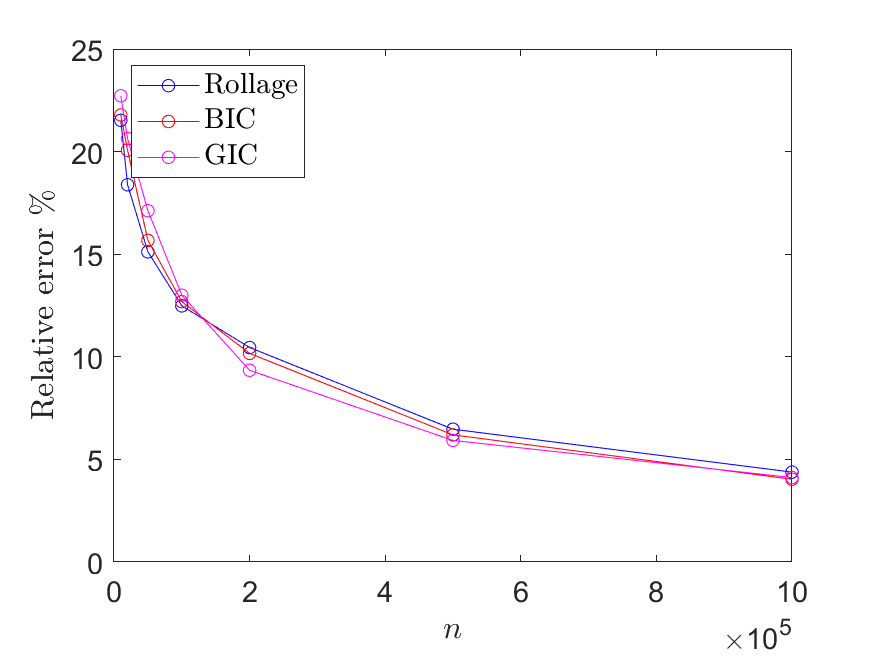}
		\caption{$\mathtt{ARMA}(75,25)$}
	\end{subfigure}
	\caption{The first two plots (left to right) show the values of $\tilde{p}$ suggested by the 3 criterion plotted against $n$ for two $\mathtt{ARMA}$ models. The last two plots show the $Relative$ $Error$ $\%$ of the full $\mathtt{ARMA}$ parameter estimates for the same models when the optimal $\tilde{p}$ according to each criterion was used in Durbin's Algorithm (\ref{alg:hr}).}
	\label{fig:pt_vs_n_relerr}
\end{figure}


\begin{table}[h!]
\centering
\begin{tabular}{ccccccc}
\hline
\rowcolor[HTML]{EFEFEF} 
 & \textbf{$p$} & \textbf{$q$} & \textbf{$\log(n)$} & \textbf{$p\log(n)$} & \textbf{$q\log(n)$} & \textbf{$R^2_a$} \\ \hline
\cellcolor[HTML]{EFEFEF}\textbf{\texttt{Rollage}} & 2.08 & -6.29 & 4.19 & -0.14 & 0.76 & 0.92 \\ \hline
\cellcolor[HTML]{EFEFEF}BIC & 1.87 & -8.12 & 3.55 & -0.11 & 0.94 & 0.94 \\ \hline
\cellcolor[HTML]{EFEFEF}GIC & 2.54 & -8.86 & 7.09 & -0.19 & 1.03 & 0.90 \\ \hline
\end{tabular}
\caption{Coefficients for the corresponding terms of linear models fitted to predict the optimal $\tilde{p}$ chosen by each criterion based on all $\mathtt{ARMA}$ simulated time series data.}
\label{tab:linear_models}
\end{table}

\begin{figure}[h!]
	\centering
		\begin{subfigure}{.24\textwidth}
		\includegraphics[width=\textwidth]{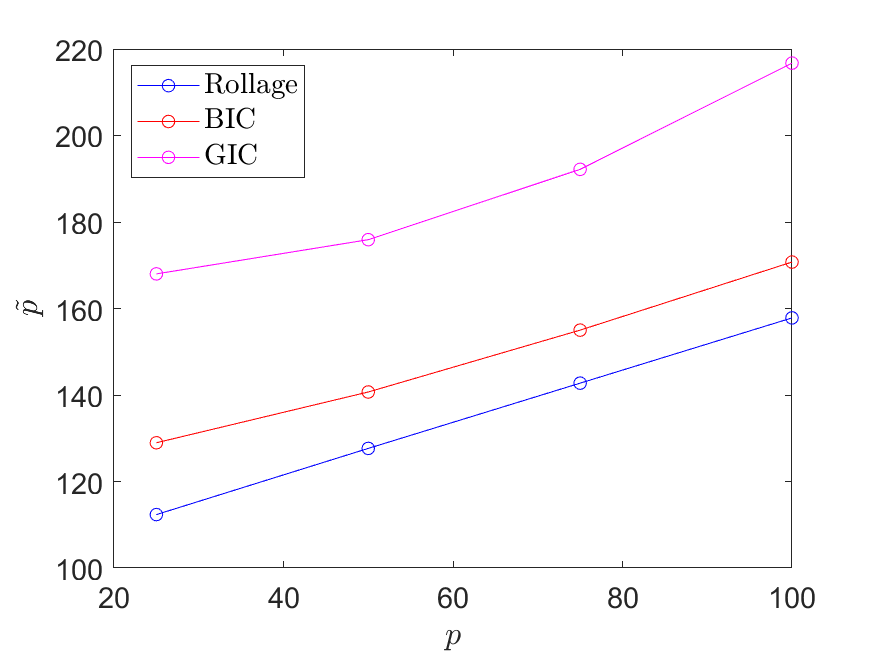}
		\caption{$\mathtt{ARMA}(p,25)$}
	\end{subfigure}
	\begin{subfigure}{.24\textwidth}
		\includegraphics[width=\textwidth]{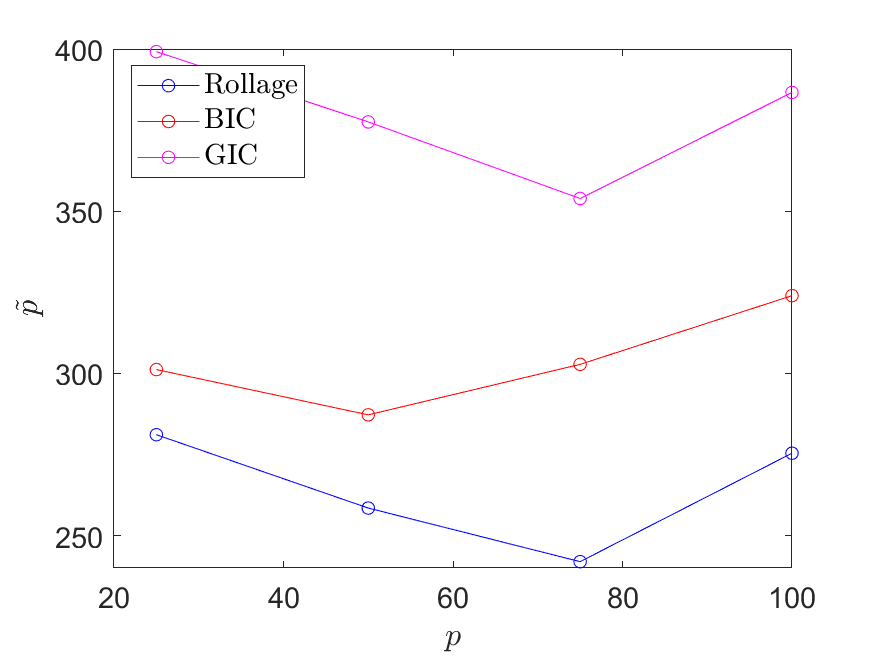}
		\caption{$\mathtt{ARMA}(p,50)$}
	\end{subfigure}
	\begin{subfigure}{.24\textwidth}
		\includegraphics[width=\textwidth]{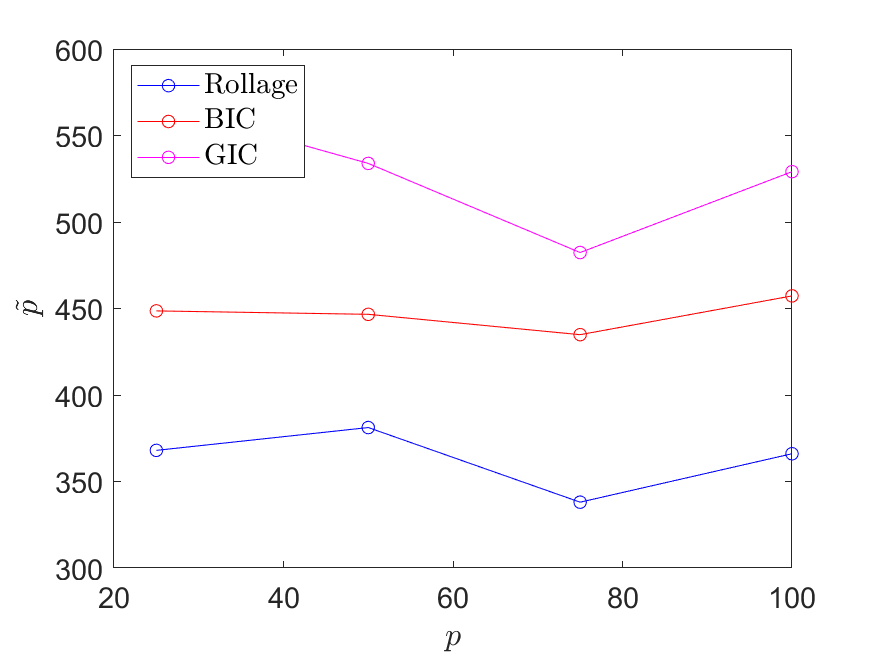}
		\caption{$\mathtt{ARMA}(p,75)$}
	\end{subfigure}
	\begin{subfigure}{.24\textwidth}
		\includegraphics[width=\textwidth]{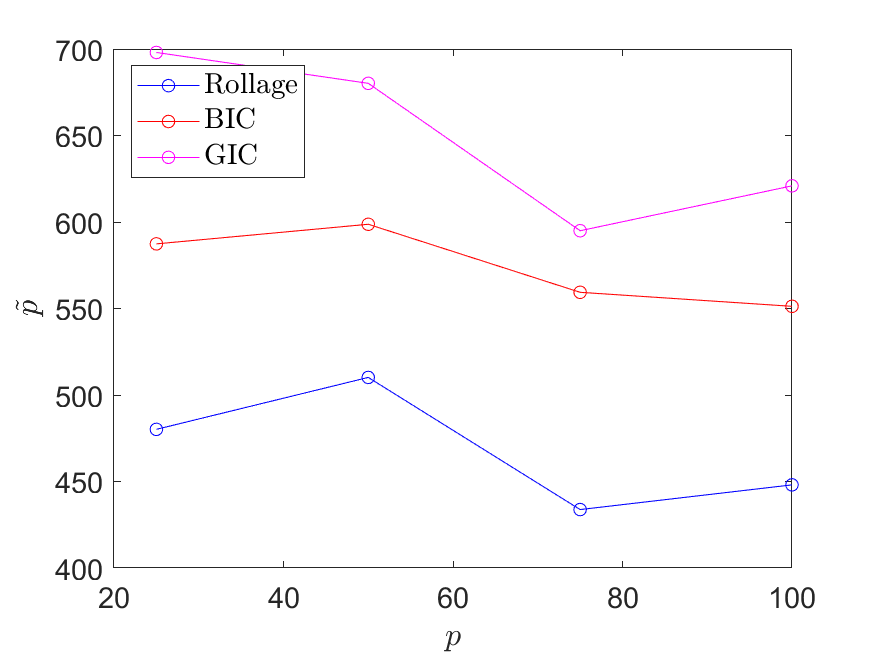}
		\caption{$\mathtt{ARMA}(p,100)$}
	\end{subfigure}
		\begin{subfigure}{.24\textwidth}
		\includegraphics[width=\textwidth]{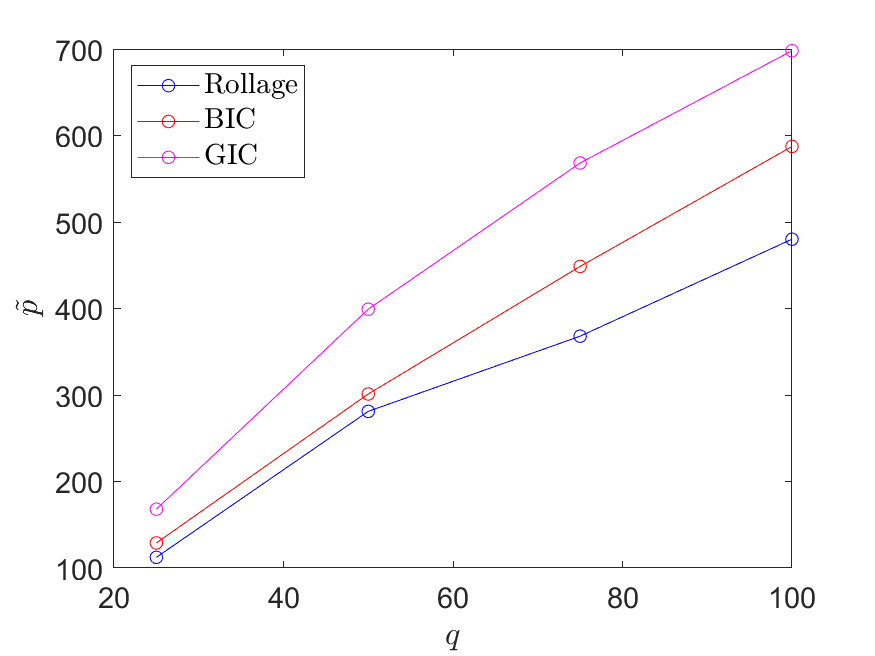}
		\caption{$\mathtt{ARMA}(25,q)$}
	\end{subfigure}
	\begin{subfigure}{.24\textwidth}
		\includegraphics[width=\textwidth]{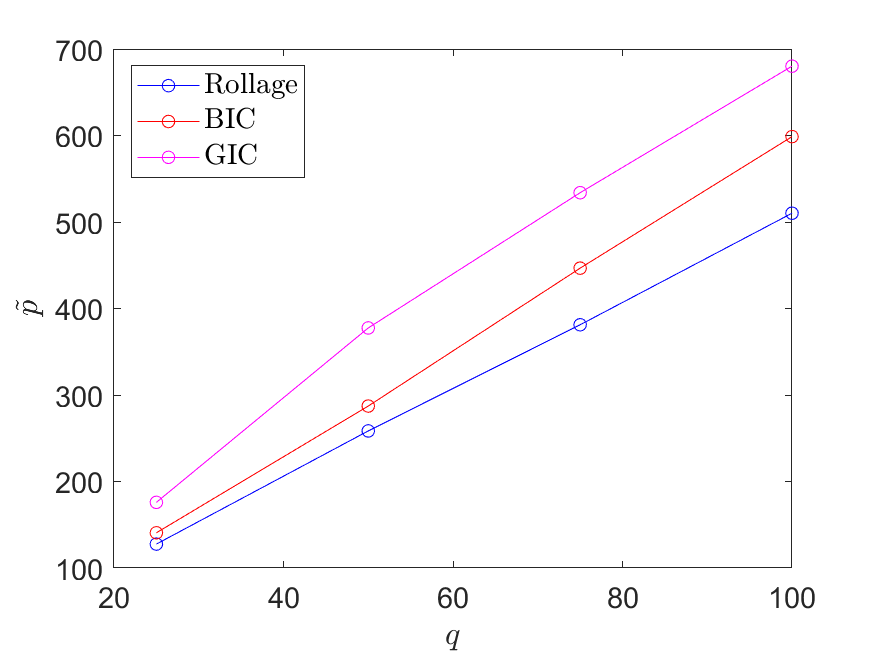}
		\caption{$\mathtt{ARMA}(50,q)$}
	\end{subfigure}
	\begin{subfigure}{.24\textwidth}
		\includegraphics[width=\textwidth]{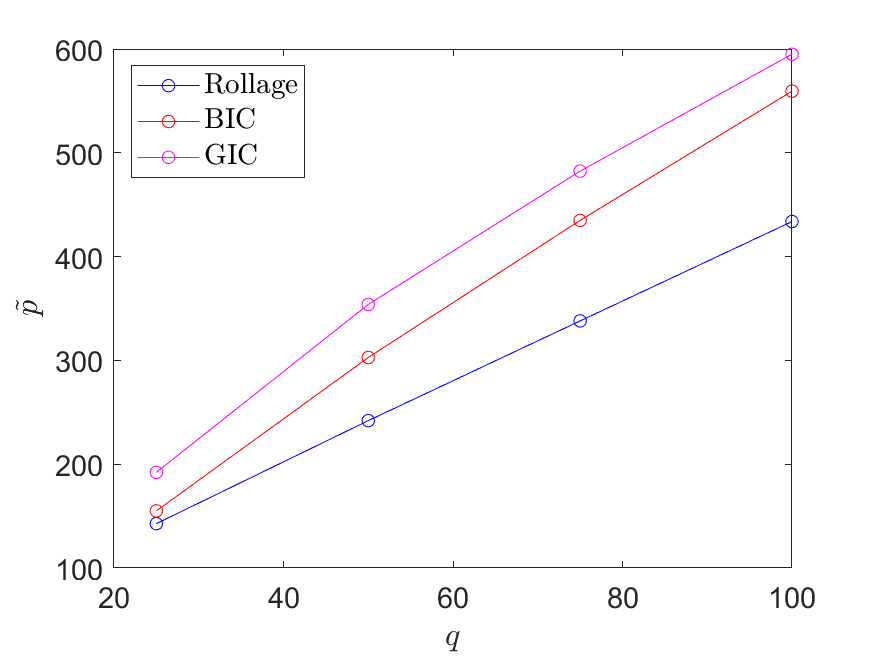}
		\caption{$\mathtt{ARMA}(75,q)$}
	\end{subfigure}
	\begin{subfigure}{.24\textwidth}
		\includegraphics[width=\textwidth]{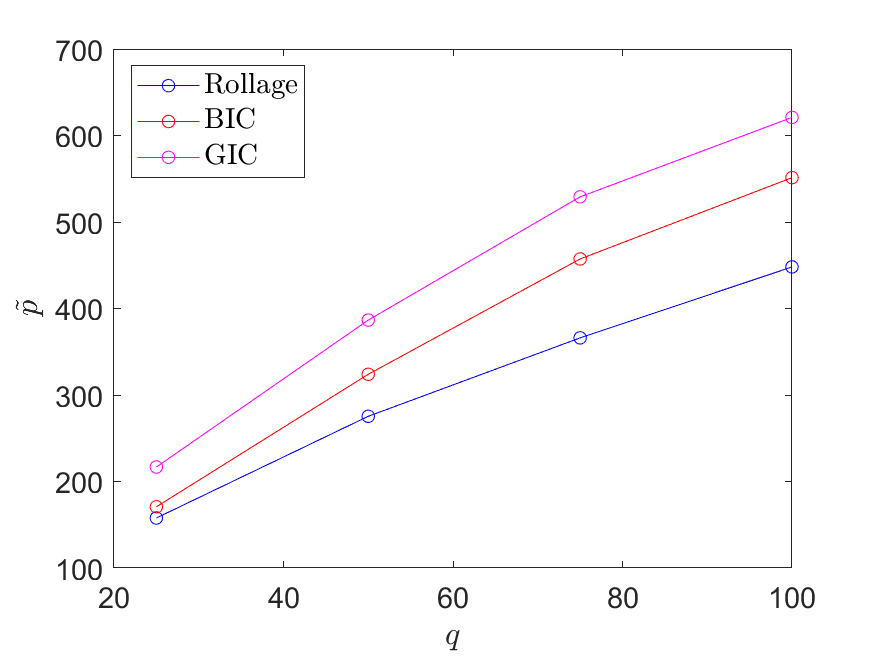}
		\caption{$\mathtt{ARMA}(100,q)$}
	\end{subfigure}
	\caption{The values of $\tilde{p}$ suggested by the three criterion are plotted against $p$ and $q$ for $n=1,000,000$. The dependence of $\tilde{p}$ linearly on $q$ is clearly more significant than the dependence on $p$.}
	\label{fig:pt_vs_p_q}
\end{figure}

	
	\section{Conclusion} \label{sec:Conclusion}

In this paper, we have developed a new efficient algorithm to estimate an \texttt{AR} model. Motivated by \cref{thm:lim_dist_MLE}, we utilise the concept of a rolling average to develop an algorithm, called \texttt{Rollage}, to estimate the order of an \texttt{AR} model and subsequently fit the model. When used in conjunction with existing methodology, specifically Durbin's, the \texttt{Rollage} algorithm can be used as a criterion to optimally fit \texttt{ARMA} models to big time series data. Empirical results on large-scale synthetic time series data show the efficacy of \texttt{Rollage} and when compared to existing criteria, the \texttt{Rollage} algorithm provides a great trade off between computational complexity and algorithmic accuracy.

	
	
	

	\clearpage
	\appendix
\section{Technical Lemmas/Propositions/Theorems and Proofs}
\label{sec:proofs}


\subsection{Proof of \cref{thm:lrc_matrix}}

We first present \cref{lem:inv_block_matrix} \cite{Golub1983Matrix}, \cref{thm:recursion_gamma} \cite{Brockwell2009TimeSeries}, and \cref{thm:invA_detA} \cite{Golub1983Matrix} and then prove \cref{pro:sigma_matrix,pro:det_Gamma} which are used in the proof of \cref{thm:lrc_matrix}. 

\begin{lemma}[Block Matrix Inversion \cite{Golub1983Matrix}] \label{lem:inv_block_matrix}
	Consider the $2\times 2$ block matrix
	\begin{align*}
		\medskip M & = \begin{pmatrix}
			\medskip c & \bm{b}^{\transpose} \\
			\medskip \bm{b} & \AA 
		\end{pmatrix},
	\end{align*}
	where $\AA$, $\bm{b}$, and $c$ are an $m\times m$ matrix, an $m\times 1$ vector and a scalar, respectively. If $\AA$ is invariable, the inverse of matrix $M$ exists an can be calculated as follows
	\begin{align*}
		\medskip M^{-1} & = \frac{1}{k}\begin{pmatrix}
			\medskip 1 & -\bm{b}^{\transpose} \AA^{-1} \\
			\medskip -\AA^{-1} \bm{b} & k \AA^{-1} + \AA^{-1} \bm{b} \bm{b}^{\transpose} \AA^{-1}
		\end{pmatrix},
	\end{align*}
	where $k = c - \bm{b}^{\transpose} \AA^{-1} \bm{b}$.
\end{lemma}

\begin{theorem}[Recursion for Autocovariance Fucntion \cite{Brockwell2009TimeSeries}] \label{thm:recursion_gamma}
	For a causal $\mathtt{AR}(p)$ model with an induced Gaussian white noise series with mean $0$ and variance $\sigma_W^2$, the autocovariance function at lag $j$ is given by 
	\begin{align*}
		\medskip \gamma_j & = \begin{cases}
			\phi_1^{(p)} \gamma_{1} + \cdots \phi_p^{(p)} \gamma_{p} + \sigma_W^2 & \mbox{for } j=0, \\
			\phi_1^{(p)} \gamma_{j-1} + \cdots + \phi_p^{(p)} \gamma_{j-p} & \mbox{for } j=1,2,\ldots, \\
			\gamma_{-j} & \mbox{for } j=-1,-2,\ldots.
		\end{cases}
	\end{align*}
\end{theorem}

\begin{theorem}[Inverse and Adjugate of a Square Matrix \cite{Golub1983Matrix}] \label{thm:invA_detA}
	If $A$ is an invertible square matrix, its inverse can be represented by
	\begin{align*}
	\medskip A^{-1} & = \frac{1}{\mathsf{det}(A)}\mathsf{adj}(A),
	\end{align*} 
	where $\mathsf{adj}(A)$ is the adjugate of matrix $A$, that is, the transpose of its cofactor matrix.
\end{theorem}

\begin{proposition}[Recursion For Matrix $\Sigma_{p,m}$] \label{pro:sigma_matrix}
	Let the time series $\{Y_1,\cdots, Y_n\}$ be a causal $\mathtt{AR}(p)$ model. The covariance matrix $\Sigma_{p,m}$ for a fixed $p \geq 1$ and $m=p+2,p+3,\ldots$, satisfies the following recursion
	\begin{align*}
	\medskip \Sigma_{p,m} & = \begin{pmatrix}
	\medskip 1 & \bm{v}_{p,m-1}^\transpose \\
	\medskip \bm{v}_{p,m-1} & \Sigma_{p,m-1} + V_{p,m-1}
	\end{pmatrix},
	\end{align*}
	where $\bm{v}$ is an $(m-1)\times 1$ vector given by
	\begin{align*}	
	\medskip \bm{v}_{p,m-1}^{\transpose} := \begin{bmatrix}
	\phi_1^{(p)} & \cdots & \phi_p^{(p)} & 0 & \cdots & 0
	\end{bmatrix}
	\end{align*}
	and $V_{p,m-1}$ is a symmetric $(m-1)\times (m-1)$ matrix defined by
	\begin{align*}
	\medskip V_{p,m-1} & := \bm{v}_{p,m-1} \bm{v}_{p,m-1}^\transpose.
	\end{align*}
\end{proposition}
\begin{proof}
	From \cref{thm:lim_dist_MLE}, we know that $\Sigma_{p,m} = \sigma_W^2 \Gamma_{p,m}^{-1}$. Furthermore, we have
	\begin{align*}
		\medskip \Gamma_{p,m} & = \begin{pmatrix}
			\gamma_0 & \gamma_1 & \cdots & \gamma_{m-1} \\
			\gamma_1 & \gamma_0 & \cdots & \gamma_{m-2} \\
			\vdots   & \vdots & \ddots & \vdots \\
			\gamma_{m-1} & \gamma_{m-2} & \cdots & \gamma_{0}
		\end{pmatrix} \\
		& = \begin{pmatrix}
			\gamma_0 & \gamma_1 & \cdots & \gamma_{m-1} \\
			\gamma_1 & & &  \\
			\vdots   & & \Gamma_{p,m-1} & \\
			\gamma_{m-1} & & & 
		\end{pmatrix}\ =\ \begin{pmatrix}
			\gamma_0 & \bm{\gamma}_{p,m-1}^\transpose \\
			\bm{\gamma}_{p,m-1} & \Gamma_{p,m-1}
		\end{pmatrix},
	\end{align*}
	where $\bm{\gamma}_{p,m-1}^\transpose := \begin{bmatrix}
		\gamma_1 & \cdots & \gamma_{m-1}
	\end{bmatrix}$. So, \cref{lem:inv_block_matrix} implies that
	\begin{align*}
		\medskip \Gamma_{p,m}^{-1} & = \frac{1}{k_{p,m-1}}\begin{pmatrix}
			\medskip 1 & -\bm{\gamma}_{p,m-1}^{\transpose} \Gamma_{p,m-1}^{-1} \\
			\medskip -\Gamma_{p,m-1}^{-1} \bm{\gamma}_{p,m-1} & k_{p,m-1} \Gamma_{p,m-1}^{-1} + \Gamma_{p,m-1}^{-1} \bm{\gamma}_{p,m-1} \bm{\gamma}_{p,m-1}^{\transpose} \Gamma_{p,m-1}^{-1}
		\end{pmatrix},
	\end{align*}
	where $k_{p,m-1} = \gamma_0 - \bm{\gamma}_{p,m-1}^{\transpose} \Gamma_{p,m-1}^{-1} \bm{\gamma}_{p,m-1}$. Firstly, we simplify $\Gamma_{p,m-1}^{-1} \bm{\gamma}_{p,m-1}$ appearing in all blocks of the inverse matrix $\Gamma_{p,m}^{-1}$. For this purpose, let define $\bm{v}_{p,m-1} := \Gamma_{p,m-1}^{-1} \bm{\gamma}_{p,m-1}$ or equivalently, $\Gamma_{p,m-1}\bm{v}_{p,m-1} = \bm{\gamma}_{p,m-1}$. More precisely, 
	\begin{align*}
		\medskip \begin{pmatrix}
			\gamma_0 & \gamma_1 & \cdots & \gamma_{m-2} \\
			\gamma_1 & \gamma_0 & \cdots & \gamma_{m-3} \\
			\vdots   & \vdots & \ddots & \vdots \\
			\gamma_{m-2} & \gamma_{m-3} & \cdots & \gamma_{0}
		\end{pmatrix} \begin{bmatrix}
			v_1 \\ v_2 \\ \vdots \\ v_{m-1}
		\end{bmatrix} & = \begin{bmatrix}
			\gamma_1 \\ \gamma_2 \\ \vdots \\ \gamma_{m-1}
		\end{bmatrix}.
	\end{align*}
	Following \ref{thm:recursion_gamma} as well as the uniqueness of the solution of the above system of linear equations (due to the invertibility of matrix $\Gamma_{p,m-1}$), we can conclude that for any $m \geq p+2$, we have
	\begin{align*}
		\medskip \bm{v}_{p,m-1}^{\transpose} & = \begin{bmatrix}
			\phi_1^{(p)} & \cdots & \phi_p^{(p)} & 0 & \cdots & 0
		\end{bmatrix}.
	\end{align*}
	Therefore, 
	\begin{align*}
		\medskip k_{p,m-1} & = \gamma_0 - \bm{\gamma}_{p,m-1}^{\transpose} \Gamma_{p,m-1}^{-1} \bm{\gamma}_{p,m-1} \\
		\medskip & = \gamma_0 - \bm{\gamma}_{p,m-1}^{\transpose} \bm{v}_{p,m-1} \\
		\medskip & = \gamma_0 - \phi_1^{(p)}\gamma_1 - \cdots - \phi_1^{(p)}\gamma_p \\
		\medskip \mbox{(\cref{thm:recursion_gamma})} & = \sigma_W^2.
	\end{align*}
	Thus, all these results conclude that
	\begin{align*}
		\medskip \Gamma_{p,m}^{-1} & = \frac{1}{\sigma_W^{2}}\begin{pmatrix}
			\medskip 1 & -\bm{v}_{p,m-1}^{\transpose} \\
			\medskip -\bm{v}_{p,m-1} & \sigma_W^2 \Gamma_{p,m-1}^{-1} + \bm{v}_{p,m-1} \bm{v}_{p,m-1}
		\end{pmatrix},
	\end{align*}
	and, equivalently,
	\begin{align*}
		\medskip \Sigma_{p,m} & = \begin{pmatrix}
			\medskip 1 & -\bm{v}_{p,m-1}^{\transpose} \\
			\medskip -\bm{v}_{p,m-1} & \Sigma_{p,m-1} + V_{p,m-1} ,
		\end{pmatrix}.
	\end{align*}
\end{proof}

\begin{proposition}[Recursion for Determinant of Autocovariance Matrix]\label{pro:det_Gamma}
	Let the time series $\{Y_1,\cdots, Y_n\}$ be a causal $\mathtt{AR}(p)$ model with an induced Gaussian white noise series with mean $0$ and variance $\sigma_W^2$. We have the following recursion for the determinant of the autocovariance matrix:
	\begin{align*}
		\medskip \mathsf{det}(\Gamma_{p,m}) & = \sigma_W^{2} \mathsf{det}(\Gamma_{p,m-1}),\quad \mbox{for\ } p \geq 1,\, m \geq p+1.
	\end{align*}
\end{proposition}
\begin{proof}
	The autocovariance matrix $\Gamma_{p,m}$ for some $m > p$ is given by
	\begin{align*}
		\medskip \Gamma_{p,m} & = \begin{pmatrix}
			\gamma_0 & \cdots & \gamma_{m-2} & \gamma_{m-1} \\
			\vdots   & \ddots & \vdots       & \vdots \\
			\gamma_{m-2} & \cdots & \gamma_{0} & \gamma_{1} \\
			\gamma_{m-1} & \cdots & \gamma_{1} & \gamma_{0}
		\end{pmatrix}\ =\ \begin{pmatrix}
			&  & & \gamma_{m-1} \\
			& \Gamma_{p,m-1}  & & \vdots \\
			& & & \gamma_1 \\
			\gamma_{m-1} & \cdots & \gamma_{1} & \gamma_{0}
		\end{pmatrix}.
	\end{align*}
	Let $\Gamma_{p,m}(:,j)$ for $j=1,\ldots,m$ denote the $j\th$ column of the autocovariance matrix $\Gamma_{p,m}$. Motivated from \cref{thm:recursion_gamma}, performing the following column operation on the last column of $\Gamma_{p,m}$ does not change the determinant:
	\begin{align*}
		\medskip \Gamma_{p,m}(:,m) & = \Gamma_{p,m}(:,m) - \phi_{1} \Gamma_{p,m}(:,m-1) - \cdots - \phi_{p} \Gamma_{p,m}(:,m-p).
	\end{align*}
	This results in 	
	\begin{align*}
		\medskip \mathsf{det}(\Gamma_{p,m}) & = \begin{vmatrix}
			&  & & \gamma_{m-1} - \phi_{1}\gamma_{m-2} - \cdots - \phi_{p}\gamma_{m-p-1} \\
			& \Gamma_{p,m-1}  & & \vdots \\
			& & &  \gamma_{1} - \phi_{1}\gamma_{0} - \cdots - \phi_{p}\gamma_{p-1} \\
			\gamma_{m-1}  & \cdots & \gamma_{1} & \gamma_{0} - \phi_{1}\gamma_{1} - \cdots - \phi_{p}\gamma_{p}
		\end{vmatrix}.
	\end{align*}
	Applying \cref{thm:recursion_gamma} to column $m$ implies that, 
	\begin{align*}
		\medskip \mathsf{det}(\Gamma_{p,m}) & = \begin{vmatrix}
			&  &  & 0 \\
			& \Gamma_{p,m-1} &  & \vdots \\
			& & & 0 \\
			\gamma_{m-1}  & \cdots & \gamma_{1} & \sigma_W^2
		\end{vmatrix} \\
		\medskip & = (-1)^{m+m}\sigma_W^2 \mathsf{det}(\Gamma_{p,m-1}) \\
		\medskip & = \sigma_W^2 \mathsf{det}(\Gamma_{p,m-1}).		
	\end{align*}
\end{proof}

By considering all these results together, we can now prove \cref{thm:lrc_matrix}.

\subsubsection*{Proof of \cref{thm:lrc_matrix}}
\begin{proof}
	We prove the nested equation for the matrix $\mathsf{NLRC}_{p,m}$ by fixing $p$ and varying $m = p+1,p+2,\ldots$. For the initial condition at $ m=p+1 $, we have
	\begin{align*}
		\medskip \mathsf{NLRC}_{p,p+1}(1,1) & = \mathsf{Var}(\sqrt{n}\hat{\phi}_{p+1}^{(p+1)})\ =\ \mathsf{Var}(\sqrt{n}\widehat{\mathtt{PACF}}_{1}),
	\end{align*}
	where $\widehat{\mathtt{PACF}}_{1}$ is the sample \textsf{PACF} at lag $1$. Since the variance of the sample \textsf{PACF}, scaled by $\sqrt{n}$, at any lag is equal to $1$ (cf. \cref{thm:lim_dist_MLE}), it results in
	\begin{align*}
		\medskip \mathsf{NLRC}_{p,p+1}(1,1) & = 1\ =\ (\phi_0^{(p)})^2.
	\end{align*}
	
	\noindent Now, consider a fixed value of $m > p+1$. According to \cref{def:NLRC_matrix,pro:sigma_matrix}, we have 
	\begin{align*}
		\medskip \mathsf{NLRC}_{p,m}(2:m-p, 2:m-p) & = \Sigma_{p,m}(p+2:m, p+2:m) \\
		\mbox{(\cref{pro:sigma_matrix})} & = \Sigma_{p,m-1}(p+1:m-1, p+1:m-1) \\
		\medskip & \hspace*{0.5cm} + V_{p,m-1}(p+1:m-1, p+1:m-1) \\
		\medskip & = \Sigma_{p,m-1}(p+1:m-1, p+1:m-1) \\
		\medskip & = {\mathsf{NLRC}}_{p,m-1}.
	\end{align*}
	Note that the second last equality is due to the structure of matrix $V_{p,m-1}$ as defined in \cref{pro:sigma_matrix}. More precisely, except the $p\times p$ block on the top left corner of this matrix, all the other coordinates are equal to zero, including the $(m-p-1)\times (m-p-1)$ block on the lower right corner, that is $V_{p,m-1}(p+1:m-p-1, p+1:m-p-1)$. 
	
	\noindent To complete this proof, we only need to show that the nested equations provided for the first column of $\mathsf{NLRC}_{p,m}$ hold. For this purpose, we utilize \cref{thm:lim_dist_MLE,thm:invA_detA} in the proof. For a fixed $i\in\{ 1,\ldots,m-p \}$, we have
	\begin{align}
		\medskip \nonumber \mathsf{NLRC}_{p,m}(i,1) & = \Sigma_{p,m}(p+i,p+1) \\ 
		\medskip \nonumber \mbox{(\cref{thm:lim_dist_MLE})} & = \sigma_W^2 \Gamma_{p,m}^{-1}(p+i,p+1) \\
		\medskip \label{eq:thm:lrc_matrix_1} \mbox{(\cref{thm:invA_detA})} & = \sigma_W^2\frac{\mathsf{adj}(\Gamma_{p,m})(p+i,p+1)}{\mathsf{det}(\Gamma_{p,m})}.
	\end{align}
	The value of $\mathsf{adj}(\Gamma_{p,m})(p+i,p+1)$ is equal to the coordinate $(p+1,p+i)$ of the cofactor matrix of $\Gamma_{p,m}$. However, as the latter matrix is symmetric, its cofactor matrix is symmetric as well, implying that both coordinates $(p+1,p+i)$ and $(p+i,p+1)$ of the cofactor matrix are identical. This results in
	\begin{align*}
		\medskip \mathsf{adj}(\Gamma_{p,m})(p+i,p+1) & = (-1)^{2p+i+1} \begin{vmatrix}
			\gamma_0 & \cdots & \gamma_{p-1} & \gamma_{p+1} & \cdots & \gamma_{m-1} \\
			\vdots   & \ddots & \vdots & \vdots &  \ddots & \vdots \\
			\gamma_{p+i-2} & \cdots & \gamma_{i-1} & \gamma_{i-3} & \cdots & \gamma_{m-p-i+1} \\
			\gamma_{p+i} & \cdots & \gamma_{i+1} & \gamma_{i-1} & \cdots & \gamma_{m-p-i-1} \\
			\vdots   & \ddots & \vdots & \vdots &  \ddots & \vdots \\
			\gamma_{m-1} & \cdots & \gamma_{m-p} &  \gamma_{m-p-2} & \cdots &\gamma_{0}
		\end{vmatrix}.
	\end{align*}
	We find this determinant in four cases and consequently derive $\mathsf{NLRC}_{p,m}(i,1)$ for each case:
	
	\begin{itemize}
		\item Case $1$: $1 \leq i \leq m-p-1$ and $m \leq 2p+1$. Similar to the proof of \cref{pro:det_Gamma}, we can add a weighted sum of some columns to the last column and then utilize \cref{thm:recursion_gamma} to simplify the result while keeping the determinant unchanged. However, since $m \leq 2p+1$ and the $(p+1)\st$ column of the matrix $\Sigma_{p,m}$ has been already removed, implementing such operations on the last column will result in the coordinates of the removed column multiplied by its corresponding coefficient as in \cref{thm:recursion_gamma}. More precisely, by adding to the last column the $p-1$ preceding columns multiplied by $-\phi_1^{(p)}, \cdots, -\phi_{m-p-2}^{(p)}, -\phi_{m-p}^{(p)}, \cdots, -\phi_p^{(p)}$, respectively, and applying \cref{thm:recursion_gamma}, we yield
		\begin{align*}
			\medskip & \mathsf{adj}(\Gamma_{p,m})(p+i,p+1) = 	(-1)^{i+1}\begin{vmatrix}
				\gamma_0 & \cdots & \gamma_{p-1} & \gamma_{p+1} & \cdots & 	\phi_{m-p-1}^{(p)}\gamma_{p} \\
				\vdots   & \ddots & \vdots & \vdots &  \ddots & \vdots \\
				\gamma_{p+i-2} & \cdots & \gamma_{i-1} & \gamma_{i-3} & \cdots & \phi_{m-p-1}^{(p)}\gamma_{i-2} \\
				\gamma_{p+i} & \cdots & \gamma_{i+1} & \gamma_{i-1} & \cdots & 	\phi_{m-p-1}^{(p)}\gamma_{i} \\
				\vdots   & \ddots & \vdots & \vdots &  \ddots & \vdots \\
				\gamma_{m-1} & \cdots & \gamma_{m-p} &  \gamma_{m-p-2} & \cdots 	& \phi_{m-p-1}^{(p)}\gamma_{m-p-1} + \sigma_W^2
			\end{vmatrix} \\
			\medskip & \hspace*{0.1cm} = (-1)^{m-p+i-1}\phi_{m-p-1}^{(p)} \begin{vmatrix}
				\gamma_0 & \cdots & \gamma_{p-1} & \gamma_{p} & \gamma_{p+1} & \cdots & \gamma_{m-2} \\
				\vdots   & \ddots & \vdots & \vdots &  \ddots & \vdots \\
				\gamma_{p+i-2} & \cdots & \gamma_{i-1} & \gamma_{i-2} & \gamma_{i-3} & \cdots & \gamma_{m-p-i} \\
				\gamma_{p+i} & \cdots & \gamma_{i+1} & \gamma_{i} & \gamma_{i-1} & \cdots & \gamma_{m-p-i-2} \\
				\vdots   & \ddots & \vdots & \vdots &  \ddots & \vdots \\
				\gamma_{m-1} & \cdots & \gamma_{m-p} &  \gamma_{m-p-1} + \sigma_W^2/\phi_{m-p-1}^{(p)} & \gamma_{m-p-2} & \cdots & \gamma_1
			\end{vmatrix}.
		\end{align*}
		In the last equality, the determinant is multiplied by $(-1)^{m-p-2}$ due to the $m-p-2$ column exchange operations to bring the last column to the $(p+1)\st$ column. Now, by implementing a similar operation on the last row we obtain 
		\begin{align*}
			\medskip & \mathsf{adj}(\Gamma_{p,m})(p+i,p+1) = (-1)^{m-p+i-1}\phi_{m-p-1}^{(p)} \\
			\medskip & \hspace*{0.1cm} \times\begin{vmatrix}
				\gamma_0 & \cdots & \gamma_{p-1} & \gamma_{p} & \gamma_{p+1} & \cdots & \gamma_{m-2} \\
				\vdots   & \ddots & \vdots & \vdots &  \ddots & \vdots \\
				\gamma_{p+i-2} & \cdots & \gamma_{i-1} & \gamma_{i-2} & \gamma_{i-3} & \cdots & \gamma_{m-p-i} \\
				\gamma_{p+i} & \cdots & \gamma_{i+1} & \gamma_{i} & \gamma_{i-1} & \cdots & \gamma_{m-p-i-2} \\
				\vdots   & \ddots & \vdots & \vdots &  \ddots & \vdots \\
				\phi_{m-p-i}^{(p)}\gamma_{p+i-1} & \cdots &\phi_{m-p-i}^{(p)} \gamma_{i} &  \phi_{m-p-i}^{(p)}\gamma_{i-1} + \sigma_W^2/\phi_{h-p}^{(p)} & \phi_{m-p-i}^{(p)}\gamma_{i-2} & \cdots & \phi_{m-p-i}^{(p)}\gamma_{m-p-i-1}
			\end{vmatrix} \\
			\medskip & \hspace*{1cm} = (-1)^{(m-p+i-1) + (m-p-i-1)}\phi_{m-p-1}^{(p)}\phi_{m-p-i}^{(p)} \\
			\medskip & \hspace*{1.5cm} \times\begin{vmatrix}
				\gamma_0 & \cdots & \gamma_{p-1} & \gamma_{p} & \gamma_{p+1} & \cdots & \gamma_{m-2} \\
				\vdots   & \ddots & \vdots & \vdots &  \ddots & \vdots \\
				\gamma_{p+i-2} & \cdots & \gamma_{i-1} & \gamma_{i-2} & \gamma_{i-3} & \cdots & \gamma_{m-p-i} \\
				\gamma_{p+i-1} & \cdots & \gamma_{i} &  \gamma_{i-1} + \sigma_W^2/(\phi_{m-p-1}^{(p)}\phi_{m-p-i}^{(p)}) & \gamma_{i-2} & \cdots & \gamma_{m-p-i-1} \\
				\gamma_{p+i} & \cdots & \gamma_{i+1} & \gamma_{i} & \gamma_{i-1} & \cdots & \gamma_{m-p-i-2} \\
				\vdots   & \ddots & \vdots & \vdots &  \ddots & \vdots \\
				\gamma_{m-2} & \cdots & \gamma_{m-p-1} &  \gamma_{m-p-2} & \gamma_{m-p-3} & \cdots & \gamma_{0}
			\end{vmatrix} \\
			\medskip & \hspace*{1cm} = \phi_{m-p-1}^{(p)}\phi_{m-p-i}^{(p)} \begin{vmatrix}
				\gamma_0 & \cdots & \gamma_{p-1} & \gamma_{p} & \gamma_{p+1} & \cdots & \gamma_{m-2} \\
				\vdots   & \ddots & \vdots & \vdots &  \ddots & \vdots \\
				\gamma_{p+i-2} & \cdots & \gamma_{i-1} & \gamma_{i-2} & \gamma_{i-3} & \cdots & \gamma_{m-p-i} \\
				\gamma_{p+i-1} & \cdots & \gamma_{i} &  \gamma_{i-1} & \gamma_{i-2} & \cdots & \gamma_{m-p-i-1} \\
				\gamma_{p+i} & \cdots & \gamma_{i+1} & \gamma_{i} & \gamma_{i-1} & \cdots & \gamma_{m-p-i-2} \\
				\vdots   & \ddots & \vdots & \vdots &  \ddots & \vdots \\
				\gamma_{m-2} & \cdots & \gamma_{m-p-1} &  \gamma_{m-p-2} & \gamma_{m-p-3} & \cdots & \gamma_{0}
			\end{vmatrix} \\
			\medskip & \hspace*{1.5cm} + \phi_{m-p-1}^{(p)}\phi_{m-p-i}^{(p)} \begin{vmatrix}
				\gamma_0 & \cdots & \gamma_{p-1} & 0 & \gamma_{p+1} & \cdots & \gamma_{m-2} \\
				\vdots   & \ddots & \vdots & \vdots &  \ddots & \vdots \\
				\gamma_{p+i-2} & \cdots & \gamma_{i-1} & 0 & \gamma_{i-3} & \cdots & \gamma_{m-p-i} \\
				\gamma_{p+i-1} & \cdots & \gamma_{i} & \sigma_W^2/(\phi_{m-p-1}^{(p)}\phi_{m-p-i}^{(p)}) & \gamma_{i-2} & \cdots & \gamma_{m-p-i-1} \\
				\gamma_{p+i} & \cdots & \gamma_{i+1} & 0 & \gamma_{i-1} & \cdots & \gamma_{m-p-i-2} \\
				\vdots   & \ddots & \vdots & \vdots &  \ddots & \vdots \\
				\gamma_{m-2} & \cdots & \gamma_{m-p-1} &  0 & \gamma_{m-p-3} & \cdots & \gamma_{0}
			\end{vmatrix} \\
			\medskip & \hspace*{1cm} = \phi_{m-p-1}^{(p)}\phi_{m-p-i}^{(p)}\mathsf{det}(\Gamma_{p,m-1}) + \sigma_W^2 \mathsf{adj}(\Gamma_{p,m-1})(p+i,p+1).
		\end{align*}		
		This result along with \cref{eq:thm:lrc_matrix_1,pro:det_Gamma} imply that,
		\begin{align*}
			\medskip \mathsf{NLRC}_{p,m}(i,1) & = \sigma_W^2\frac{\phi_{m-p-1}^{(p)}\phi_{m-p-i}^{(p)} \mathsf{det}(\Gamma_{p,m-1}) + \sigma_W^2 \mathsf{adj}(\Gamma_{p,m-1})(p+i,p+1)}{\mathsf{det}(\Gamma_{p,m})} \\
			\medskip (\mbox{\cref{pro:det_Gamma}}) & = \phi_{m-p-1}^{(p)}\phi_{m-p-i}^{(p)} + \sigma_W^2\frac{\mathsf{adj}(\Gamma_{p,m-1})(p+i,p+1)}{\mathsf{det}(\Gamma_{p,m-1})} \\
			\medskip & = \phi_{m-p-1}^{(p)}\phi_{m-p-i}^{(p)} + \mathsf{NLRC}_{p,m-1}(i,1) \quad \mbox{for } i=1,\ldots,m-p-1.
		\end{align*}

		\item Case $2$: $i = m-p$ and $m \leq 2p+1$. Analogous to Case $1$, we have
		\begin{align*}
			\medskip & \mathsf{adj}(\Gamma_{p,m})(m,p+1) = (-1)^{m-p+1} \begin{vmatrix}
				\gamma_0 & \cdots & \gamma_{p-1} & \gamma_{p+1} & \cdots & \gamma_{m-1} \\
				\gamma_1 & \cdots & \gamma_{p-2} & \gamma_{p} & \cdots & \gamma_{m-2} \\
				\vdots   & \ddots & \vdots & \vdots &  \ddots & \vdots \\
				\gamma_{m-2} & \cdots & \gamma_{m-p-1} &  \gamma_{m-p-3} & \cdots & \gamma_{1}
			\end{vmatrix} \\
			\medskip & \hspace*{1cm} = (-1)^{m-p+1} \begin{vmatrix}
					\gamma_0 & \cdots & \gamma_{p-1} & \gamma_{p+1} & \cdots & 	\phi_{m-p-1}^{(p)}\gamma_{p} \\
					\gamma_1 & \cdots & \gamma_{p-2} & \gamma_{p} & \cdots & 	\phi_{m-p-1}^{(p)}\gamma_{p-1} \\
					\vdots   & \ddots & \vdots & \vdots &  \ddots & \vdots \\
					\gamma_{m-2} & \cdots & \gamma_{m-p-1} &  \gamma_{m-p-3} & \cdots 	& \phi_{m-p-1}^{(p)}\gamma_{m-p-2}
				\end{vmatrix} \\
			\medskip & \hspace*{1cm} = (-1)^{(m-p+1) + (m-p-2)}\phi_{m-p-1}^{(p)} \begin{vmatrix}
				\gamma_0 & \cdots & \gamma_{p-1} & \gamma_{p} & \gamma_{p+1} & \cdots & \gamma_{m-2} \\
				\gamma_1 & \cdots & \gamma_{p-2} & \gamma_{p-1} & \gamma_{p} & \cdots & \gamma_{m-3} \\
				\vdots   & \ddots & \vdots & \vdots &  \ddots & \vdots \\
				\gamma_{m-2} & \cdots & \gamma_{m-p-1} &  \gamma_{m-p-2} & \gamma_{m-p-3} & \cdots & \gamma_0
			\end{vmatrix} \\
			\medskip & \hspace*{1cm} = \phi_{0}^{(p)}\phi_{m-p-1}^{(p)}\mathsf{det}(\Gamma_{p,m-1}).
		\end{align*}
		This result along with \cref{eq:thm:lrc_matrix_1,pro:det_Gamma} imply that,
		\begin{align*}
			\medskip \mathsf{NLRC}_{p,m}(m-p,1) & = \sigma_W^2\frac{\phi_{0}^{(p)}\phi_{m-p-1}^{(p)}\mathsf{det}(\Gamma_{p,m-1})}{\mathsf{det}(\Gamma_{p,m})} \\
			\medskip (\mbox{\cref{pro:det_Gamma}}) & = \phi_{0}^{(p)}\phi_{m-p-1}^{(p)}.
		\end{align*}

		\item Case $3$: $1 \leq i \leq m-p-1$ and $m > 2p+1$. Unlike Case $1$, in this case all preceding $p$ columns of the last column exist in the determinant. So, analogous to the proof of \cref{pro:det_Gamma}, we add to the last column the $p$ preceding columns multiplied by $-\phi_1, \cdots, -\phi_p$, respectively, and apply \cref{thm:recursion_gamma} to yield
		\begin{align*}
			\medskip \mathsf{adj}(\Gamma_{p,m})(p+i,p+1) & = 	(-1)^{i+1}\begin{vmatrix}
				\gamma_0 & \cdots & \gamma_{p-1} & \gamma_{p+1} & \cdots & 	0 \\
				\vdots   & \ddots & \vdots & \vdots &  \ddots & \vdots \\
				\gamma_{p+i-2} & \cdots & \gamma_{i-1} & \gamma_{i-3} & \cdots & 0 \\
				\gamma_{p+i} & \cdots & \gamma_{i+1} & \gamma_{i-1} & \cdots & 	0 \\
				\vdots   & \ddots & \vdots & \vdots &  \ddots & \vdots \\
				\gamma_{m-1} & \cdots & \gamma_{m-p} &  \gamma_{m-p-2} & \cdots 	& \sigma_W^2
			\end{vmatrix} \\
			\medskip & = (-1)^{(i+1) + (2(m-1))}\sigma_W^{2} \begin{vmatrix}
				\gamma_0 & \cdots & \gamma_{p-1} & \gamma_{p+1} & \cdots & 	\gamma_{m-2} \\
				\vdots   & \ddots & \vdots & \vdots &  \ddots & \vdots \\
				\gamma_{p+i-2} & \cdots & \gamma_{i-1} & \gamma_{i-3} & \cdots & \gamma_{m-p-i} \\
				\gamma_{p+i} & \cdots & \gamma_{i+1} & \gamma_{i-1} & \cdots & \gamma_{m-p-i-2} \\
				\vdots   & \ddots & \vdots & \vdots &  \ddots & \vdots \\
				\gamma_{m-2} & \cdots & \gamma_{m-p-1} &  \gamma_{m-p-3} & \cdots 	& \gamma_0
			\end{vmatrix} \\
			\medskip & = \sigma_W^2 \mathsf{adj}(\Gamma_{p,m-1})(p+i,p+1).
		\end{align*}
		This result along with \cref{eq:thm:lrc_matrix_1,pro:det_Gamma,thm:lim_dist_MLE} imply that,
		\begin{align*}
			\medskip \mathsf{NLRC}_{p,m}(i,1) & = \sigma_W^2\frac{\sigma_W^2 \mathsf{adj}(\Gamma_{p,m-1})(p+i,p+1)}{\mathsf{det}(\Gamma_{p,m})} \\
			\medskip (\mbox{\cref{pro:det_Gamma}}) & = \sigma_W^2\frac{ \mathsf{adj}(\Gamma_{p,m-1})(p+i,p+1)}{\mathsf{det}(\Gamma_{p,m})} \\
			\medskip \mbox{(\cref{thm:lim_dist_MLE})} & = \mathsf{NLRC}_{p,m-1}(i,1)\quad \mbox{for } i=p+2,\ldots,m-p-1.
		\end{align*}
	
		\item Case $4$: $i = m-p$ and $m > 2p+1$. Analogous to Cases $2$ and $3$, we obtain 
		\begin{align*}
			\medskip \mathsf{adj}(\Gamma_{p,m})(m,p+1) & = (-1)^{m-p+1} \begin{vmatrix}
				\gamma_0 & \cdots & \gamma_{p-1} & \gamma_{p+1} & \cdots & \gamma_{m-1} \\
				\gamma_1 & \cdots & \gamma_{p-2} & \gamma_{p} & \cdots & \gamma_{m-2} \\
				\vdots   & \ddots & \vdots & \vdots &  \ddots & \vdots \\
				\gamma_{m-2} & \cdots & \gamma_{m-p-1} &  \gamma_{m-p-3} & \cdots & \gamma_{1}
			\end{vmatrix} \\
			\medskip & = (-1)^{m-p+1} \begin{vmatrix}
				\gamma_0 & \cdots & \gamma_{p-1} & \gamma_{p+1} & \cdots & 	0 \\
				\gamma_1 & \cdots & \gamma_{p-2} & \gamma_{p} & \cdots & 	0 \\
				\vdots   & \ddots & \vdots & \vdots &  \ddots & \vdots \\
				\gamma_{m-2} & \cdots & \gamma_{m-p-1} &  \gamma_{m-p-3} & \cdots 	& 0
			\end{vmatrix} \\
			\medskip & = 0.
		\end{align*}
	This result along with \cref{eq:thm:lrc_matrix_1} imply that,
	\begin{align*}
		\medskip \mathsf{NLRC}_{p,m}(m-p,1) & = 0.
	\end{align*}
	\end{itemize} 
	Now, by starting with the scaler $\mathsf{NLRC}_{p,p+1}(1,1) = 1$ and implementing the nested equations iteratively on $m$, the coordinates of the matrix $\mathsf{NLRC}_{p,m}$ are accordingly derived.
\end{proof}

\subsection{Proof of \cref{thm:lim_dist_rolling_average}}

\begin{proof}
	We prove this theorem by fixing $p$ and utilizing \cref{def:rolling_average} and \cref{thm:lrc_matrix} in three cases:
	\begin{itemize}
		\item Case $1$: $m = p+1$.
		\begin{align*}
			\medskip \sigma_{p,p+1}^2 & = \mathsf{Var}(\sqrt{n} \bar{\phi}_{p,p+1}) \\
			\medskip & = \mathsf{Var}(\sqrt{n} \hat{\phi}_{p+1}^{(p+1)}) \\
			\medskip & = \mathsf{NLRC}_{p,p+1}(1,1) \\
			\medskip & = (\phi_0^{(p)})^2.
		\end{align*} 
		
		\item Case $2$: $p+1 < m \leq 2p$.
		\begin{align*}
			\medskip (m-p)^2\sigma_{p,m}^2 & = (m-p)^2\mathsf{Var}(\sqrt{n}\bar{\phi}_{p,m}) \\
			\medskip \mbox{(\cref{def:rolling_average})} & = (m-p)^2\mathsf{Var}(\frac{1}{m-p} \sum_{i=p+1}^{m} \sqrt{n}\hat{\phi}_{i}^{(m)}) \\
			\medskip & = \mathsf{Var}(\sqrt{n}\hat{\phi}_{p+1}^{(m)}) + \mathsf{Var}(\sum_{i=p+2}^{m} \sqrt{n}\hat{\phi}_{i}^{(m)}) + 2\mathsf{Cov}(\sqrt{n}\hat{\phi}_{p+1}^{(m)}, \sum_{i=p+2}^{m} \sqrt{n}\hat{\phi}_{i}^{(m)}) \\
			 & = \mathsf{Var}(\sqrt{n}\hat{\phi}_{p+1}^{(m)}) + \left(\sum_{i=p+2}^{m}\mathsf{Var}(\sqrt{n}\hat{\phi}_{i}^{(m)}) + 2 \sum_{i=p+3}^{m}\sum_{j=p+2}^{i-1} \mathsf{Cov}(\sqrt{n}\hat{\phi}_{i}^{(m)},\sqrt{n}\hat{\phi}_{j}^{(m)})\right) \\
			\medskip & \hspace*{0.8cm} + 2\sum_{i=p+2}^{m}\mathsf{Cov}(\sqrt{n}\hat{\phi}_{p+1}^{(m)}, \sqrt{n}\hat{\phi}_{i}^{(m)}) \\
			 \mbox{(\cref{thm:lrc_matrix})}  & = \mathsf{NLRC}_{p,m}(1,1) + \left(\sum_{i=2}^{m-p}\mathsf{NLRC}_{p,m}(i,i) + 2 \sum_{i=3}^{m-p}\sum_{j=2}^{i-1} \mathsf{NLRC}_{p,m}(i,j)\right) \\
			\medskip & \hspace*{0.8cm} + 2\sum_{i=2}^{m-p}\mathsf{NLRC}_{p,m}(i,1) \\
			\mbox{(\cref{thm:lrc_matrix})} & = \left(\sum_{i=2}^{m-p}\mathsf{NLRC}_{p,m-1}(i-1,i-1) + 2 \sum_{i=3}^{m-p}\sum_{j=2}^{i-1} \mathsf{NLRC}_{p,m-1}(i-1,j-1)\right) \\
			\medskip & \hspace*{0.8cm} + \left(\mathsf{NLRC}_{p,m}(1,1) + 2\sum_{i=2}^{m-p}\mathsf{NLRC}_{p,m}(i,1)\right) \\
			 & = \left(\sum_{i=1}^{m-p-1}\mathsf{NLRC}_{p,m-1}(i,i) + 2 \sum_{i=2}^{m-p-1}\sum_{j=1}^{i-1} \mathsf{NLRC}_{p,m-1}(i,j)\right) \\
			\medskip & \hspace*{0.8cm} + \left(\mathsf{NLRC}_{p,m}(1,1) + 2\sum_{i=2}^{m-p}\mathsf{NLRC}_{p,m}(i,1)\right) \\
			 \mbox{(\cref{thm:lrc_matrix})} & = \left(\sum_{i=p+1}^{m-1}\mathsf{Var}(\sqrt{n}\hat{\phi}_{i}^{(m-1)}) + 2 \sum_{i=p+2}^{m-1}\sum_{j=p+1}^{i-1} \mathsf{Cov}(\sqrt{n}\hat{\phi}_{i}^{(m-1)},\sqrt{n}\hat{\phi}_{j}^{(m-1)})\right) \\
			\medskip & \hspace*{0.8cm} + \left(\sum_{k=0}^{m-p-1}(\phi_k^{(p)})^2 + 2\sum_{i=2}^{m-p}\sum_{k=0}^{m-p-i}\phi_k^{(p)} \phi_{k+i-1}^{(p)}\right) \\
			\medskip & = \left( \mathsf{Var}(\sum_{i=p+1}^{m-1} \sqrt{n}\hat{\phi}_{i}^{(m-1)}) \right) + \left(\sum_{k=0}^{m-p-1}(\phi_k^{(p)})^2 + 2\sum_{i=1}^{m-p-1}\sum_{k=0}^{m-p-1-i}\phi_k^{(p)} \phi_{k+i}^{(p)}\right) \\
			\medskip & = (m-p-1)^2\sigma_{p,m-1}^2 + (\phi_0^{(p)}+\cdots+\phi_{m-p-1}^{(p)})^2.
		\end{align*} 
		
		\item Case $3$: $m \geq 2p+1$. 
		The proof of this case is analogous to Case $2$. To avoid duplication, it is omitted. 
	\end{itemize}
	Now, it is readily seen that by starting from Case 1 and implementing the recursion iteratively, the provided general solution is derived. 
\end{proof}

%
%
%
%
%
%



\bibliographystyle{plain}
\bibliography{Biblio}

\end{document}